\documentclass[11pt]{article}
\usepackage{graphicx}
\usepackage{amsmath}
\usepackage{multirow}
\usepackage{amssymb}
\usepackage{amsthm}
\usepackage{natbib}
\usepackage{moreverb}
\usepackage[T1]{fontenc}
\usepackage{gensymb}
\usepackage{pdfpages}
\usepackage[figuresright]{rotating}
\usepackage{geometry}
\usepackage{bbm}
\usepackage{setspace}
\usepackage{algorithm, algpseudocode}
\usepackage{diagbox}
\usepackage{csquotes, caption}
\renewcommand{\P}{{\text{pr}}}
\newcommand{\E}{{{E}}}

\newcommand{\bx}{{\mathbf{x}}}
\newcommand{\bT}{{\mathbf{T}}}
\newcommand{\bt}{{\mathbf{t}}}

\newcommand{\bhtau}{{\boldsymbol{\hat{\tau}}}}

\newcommand{\var}{{\text{var}}}
\newcommand{\bY}{{\mathbf{Y}}}

\newcommand{\bV}{{\mathbf{V}}}
\newcommand{\1}{\mathbbm{1}}

\newcommand{\cF}{{\mathcal{F}}}

\newcommand{\cT}{{\mathcal{T}}}

\newcommand{\bq}{\mathbf{q}}

\newtheorem{theorem}{Theorem}

\newtheorem{corollary}{Corollary}
\newtheorem{lemma}{Lemma}

\theoremstyle{definition}
\newtheorem{condition}{Condition}

\theoremstyle{remark}
\newtheorem{case}{Case}
\newtheorem{remark}{Remark}

\title{Studentized sensitivity analysis for the sample average treatment effect in paired observational studies}
\author{Colin B. Fogarty \thanks{Operations Research and Statistics Group, MIT Sloan School of Management, Massachusetts Institute of Technology, Cambridge MA 02142 (e-mail: \texttt{cfogarty@mit.edu})}}
\date{}

\begin{document}

\maketitle

\begin{abstract}
A fundamental limitation of causal inference in observational studies is that perceived evidence for an effect might instead be explained by factors not accounted for in the primary analysis. Methods for assessing the sensitivity of a study's conclusions to unmeasured confounding have been established under the assumption that the treatment effect is constant across all individuals. In the potential presence of unmeasured confounding, it has been argued that certain patterns of effect heterogeneity may conspire with unobserved covariates to render the performed sensitivity analysis inadequate. We present a new method for conducting a sensitivity analysis for the sample average treatment effect in the presence of effect heterogeneity in paired observational studies. Our recommended procedure, called the studentized sensitivity analysis, represents an extension of recent work on studentized permutation tests to the case of observational studies, where randomizations are no longer drawn uniformly. The method naturally extends conventional tests for the sample average treatment effect in paired experiments to the case of unknown, but bounded, probabilities of assignment to treatment. In so doing, we illustrate that concerns about certain sensitivity analyses operating under the presumption of constant effects are largely unwarranted.
\end{abstract}
\singlespacing
\section{Introduction}
\subsection{Constant effects, then and now}
When inferring both the existence and the magnitude of causal effects, a common expedient is to assume that the effects are constant across individuals. Unease is sometimes expressed about the restrictiveness of this assumption. Indeed, the strength of the constant effects assumption (also called the assumption of additivity) was at the core of the ``Neyman-Fisher controversy,'' with Fisher favoring the sharp null that the treatment effect is zero for all individuals and Neyman recommending tests of the weaker null that the treatment effect is zero \textit{on average} for the individuals in a given experiment. \citet{ney35} suggested that inference assuming additivity in Latin Square designs could be anti-conservative in the presence of effect heterogeneity, which elicited an acerbic response by Fisher \citep{fis35comment} in which he called Neyman's understanding of the topic into question; see \citet{sab14} for a detailed discussion of the controversy and its ramifications.

The passage of time has done little to temper the debate, with both camps maintaining supporters. Those favoring Neyman's weak null focus on the seeming inadequacy of the  constant effects assumption as a description of reality.  Gelman writes that ``the presumption of constant effects corresponds to a simplified view of the world that can impede research discussion'' \citep[][p. 636]{gel15}. Advocates of Fisher's sharp null focus, among many things, on the central role of hypothesis testing in empirical falsification.  \citet{cox58} and \citet[][\S 2.4.5]{obs} discuss how rejection of the sharp null is in and of itself useful as a means of promoting future scientific inquiry, despite a rejection of the sharp null not  implying the existence of a treatment effect that is predictably positive or negative. Such a rejection may well be indicative of underlying subject-by-treatment interactions, hence identifying the existence of patterns for effects which the current experiment can neither describe nor predict. Quoting Rosenbaum, ``the variation we do not fathom today we intend to decipher tomorrow'' \citep[][p. 40]{obs}. See \citet[][\S 2]{cau17} for additional perspective  and for further quotations supporting both sides.
\subsection{Additivity in observational studies}
In the context of observational studies, the restrictiveness of the constant effects assumptions faces additional scrutiny when assessing the robustness of a study's findings to unmeasured confounding through a \textit{sensitivity analysis}. There is a perception that, with few exceptions, the methodology described in \citet[][\S 4]{obs} requires the researcher to posit a sharp null hypothesis, and that the model may not readily extend to tests of average causal effects in the face of effect heterogeneity.  Hill opines that when conducting a sensitivity analysis, ``the focus on additive treatment effects...is potentially problematic'' \citep[][p. 308]{hil02}, while Robins states that any gain from Rosenbaum's model for a sensitity analysis is offset ``...by Rosenbaum's assumption that individual outcomes are deterministic and that an additive treatment effect model holds'' \citep[][p. 310]{rob02}. There is concern, in particular, that sensitivity analyses conducted assuming constant effects may paint an overly optimistic picture of the study's sensitivity to hidden bias. The fear is that certain patterns of unmeasured confounding may conspire with the unidentified aspects of the constant effects model, rendering the analysis assuming constant effects inadequate. One particular argument is that in observational studies, individuals may self-select into the treatment group which they know to be most beneficial to them. For a treated and control individual with the same observed covariate value $x$, one may then expect a difference in the observed response between the treated and control individuals due to this ``essential heterogeneity,'' even if there truly was no average effect of the treatment at that point $x$. \citep{hec06}. A constant treatment effect model precludes varying treatment effects of this nature, leading some to call into question the utility of models assuming additivity.

\subsection{Studentization with hidden bias}
Rather than taking a stance on which null hypothesis should be preferred, this work focuses primarily on the ramifications of the debate for the interpretations ascribed to sensitivity analyses. We have in mind a practitioner who recognizes the need for conducting a sensitivity analysis when treatment assignment is beyond their control but who would ideally like the analysis to attest to the robustness of their findings in the presence of heterogeneous effects, thereby assuaging the potential fears of critics in their field. We specialize our exposition to the case of paired observational studies, and to inference conducted using the treated-minus-control difference-in-means as the test statistic. The procedure for conducting a sensitivity analysis using this test statistic and assuming constant effects within the model of \citet{ros87} is reviewed in \S\ref{sec:review}. 

After reviewing the conventional approach assuming constant effects, we assess in \S \ref{sec:largesample} whether one can construct a valid sensitivity analysis for the sample average treatment effect in the presence of effect heterogeneity. In answering this question affirmatively, we propose a natural extension of the conventional large-sample normal-based test for the sample average treatment effect in paired experiments to paired observational studies. Theorem \ref{prop:neyman} demonstrates the validity of the approach under Neyman's null in large samples; however, through its reliance on the normal for a reference distribution the test loses the nonparametric appeal of the conventional approaches for sensitivity analysis assuming constant effects. In \S \ref{sec:studentwhole}, we overcome these issues by proposing a new reference distribution based upon biased randomizations of the data. Theorems \ref{prop:student} and \ref{thm:exact} show that this reference distribution continues to provide asymptotic Type I error control under Neyman's null while, in addition, furnishing an \textit{exact} finite-sample sensitivity analysis if Fisher's sharp null is true. 

With assurance that the framework for sensitivity analysis of \citet{ros87} is compatible with heterogeneous effects, \S \ref{sec:implication} assesses whether a more traditional sensitivity analysis based on the difference-in-means is valid in the presence of heterogeneous effects.  Theorem \ref{prop:fisher} of \S \ref{sec:fail} answers this negatively, in that there exists patterns of effect heterogeneity which lead the test assuming constant effects to have an inflated Type I error rate even asymptotically. While somewhat concerning, the mere potential for anti-conservativeness does not directly answer whether sensitivity analyses valid under constant effects could yield grossly misleading perceptions of robustness to hidden bias in the presence of heterogeneous effects. Theorem \ref{prop:epsilon} of \S \ref{sec:epsilon} shows that the extent to which the traditional approach based on the difference-in-means can differ from the asymptotically valid method proposed in \S \ref{sec:studentwhole} is, loosely stated ``not by much,'' a position to which the data analysis in \S \ref{sec:example} further attests. At its core, the claim stems from the realization that even under effect heterogeneity, the sensitivity analysis based on the difference-in-means creates a candidate worst-case distribution that correctly bounds the \textit{expectation} of the test statistic's actual distribution (proven in Lemma \ref{lemma:expec}) but which may have too small a variance (shown through a numerical counterexample in \S \ref{sec:fail}). As bias trumps variance in a sensitivity analysis, the extent to which the analysis based on the conventional approach can mislead is thankfully limited.

This work reveals that while a randomization-based sensitivity analysis using the difference-in-means as a test statistic  may have the improper size in the presence of heterogeneous effects, the issue is avoided through an appropriate studentization of the test-statistic while employing the same worst-case distribution for treatment assignments. This aligns with work on robust permutation tests by \citet{jan97} and \citet{chu13} while extending the ideas to the context of potentially biased randomizations. Under no unmeasured confounding, the studentization employed is none other than that recommended by Gosset himself: the observed difference-in-means is simply divided by the conventional standard error estimator for paired studies, hence yielding the usual $t$-statistic. When hidden bias is allowed to corrupt inference, finding both the appropriate initial test statistic and the appropriate standard deviation by which to studentize is itself non-trivial, and is described in \S \ref{sec:neyman}. The method proposed in \S \ref{sec:studentwhole}, called the \textit{studentized sensitivity analysis}, combines this studentized test statistic with the worst-case distribution for treatment assignments developed in \citet{ros87} under the assumption of constant effects. The procedure operates within the familiar model for biased treatment assignments of \citet[][\S 4]{ros87, obs} and is straightforward to implement. The studentized procedure thus empowers researchers with a sensitivity analysis for the sample average treatment effect valid in the face of effect heterogeneity, while researchers having conducted a sensitivity analysis using the traditional approach can be reasonably assured that their results do not materially overstate insensitivity to hidden bias. 

\section{Sensitivity analysis for constant effects}\label{sec:review}
\subsection{Notation and review}

With few exceptions, we adopt the notation for paired studies introduced in \S 4.1 of \citet{din17} while making suitable extensions to accommodate paired observational studies. There are $n$ independent matched pairs. In the $i$th matched pair there is one individual who receives the treatment, $T_{ij} = 1$, and one who receives the control, $T_{ij'} = 0$, such that $T_{i1} + T_{i2} = 1$ for each pair. These pairs are formed on the basis of observed pre-treatment covariates $\bx_{ij}$; however, individuals may differ on the basis of an unobserved covariate $0\leq u_{ij}\leq 1$, such that $u_{i1} \neq u_{i2}$. Each individual has a potential outcome under treatment, $Y_{ij}(1)$, and under control, $Y_{ij}(0)$. Implicit in this description of the potential outcomes is the stable unit-treatment value assumption, or SUTVA \citep{rub80}. Let $\cF = \{Y_{ij}(1), Y_{ij}(0), \bx_{ij}, u_{ij}: i = 1,...,n; j = 1,2\}$ be a set containing the potential outcomes and covariates, both observed and unobserved, for the individuals in the observational study at hand. Inference moving forwards will condition upon $\cF$, such that a superpopulation model is neither assumed nor required.

The fundamental problem of causal inference is that the pair of potential outcomes $\{Y_{ij}(1), Y_{ij}(0)\}$ is not jointly observable, and hence we cannot observe the individual treatment effect $\tau_{ij} = Y_{ij}(1) - Y_{ij}(0)$ for any individual \citep{hol86}. Instead, we observe the response $Y_{ij}^{obs} = T_{ij}Y_{ij}(1) +(1-T_{ij})Y_{ij}(0)$. See \citet{ney23} and \citet{rub74} for more on the potential outcomes framework.  Let $\bar{\tau}_i = (\tau_{i1}+\tau_{i2})/2$ be the average of the two treatment effects in pair $i$. The $i$th treated-minus-control paired difference, $\hat{\tau}_i$, is
\begin{align}
\hat{\tau}_i
&= (T_{i1}-T_{i2})(Y_{i1}^{obs} - Y_{i2}^{obs}).\label{eq:pairdiff}
\end{align}Boldface will be used to represent vector quantitites; for example, $\bT = [T_{11}, T_{12}, ..., T_{n2}]$ is the vector of length $2n$ containing the treatment assignments for all individuals. Quantities dependent on the assignment vector such as $\bT$ and $\bY^{obs}$ are random, whereas $\cF$ contains quantities viewed as fixed in the forthcoming developments.

\subsection{A Model for Hidden Bias in Observational Studies}\label{sec:model}
Let $\Omega$ be the set of $2^{n}$ possible values of ${\bT}$ under the matched pairs design, i.e. $\Omega = \{\bt: t_{i1}+t_{i2} = 1, i=1,...,n\}$. In a paired randomized experiment,  each $\bt \in \Omega$ has probability $2^{-n}$ of being selected. Let $\cT$ denote the event $\bT \in \Omega$. For a paired randomized experiment $\pi_{i} = \P(T_{i1} = 1 \mid \cF, \cT ) = 1/2$, an immediate consequence of which is that $E(\hat{\tau}_i\mid \cF, \cT) = \bar{\tau}_i$ for all pairs. That is, in a paired experiment, the treated-minus-control paired difference in any pair is an unbiased estimator for the average of the two treatment effects in that pair. 

Without control over the assignment mechanism, the probabilities $\pi_{i}$ $(i= 1,...,n)$ are unknown to the researcher in an observational study. A concern in the analysis of observational studies is that $\pi_{i}\neq 1/2$ due to latent discrepancies between the unobserved covariates $u_{i1}$ and $u_{i2}$, which could in turn induce bias into $\hat{\tau}_i$ as an estimator for $\bar{\tau}_i$.  Through a sensitivity analysis, one assesses the robustness of a study's finding to deviations from a paired experiment caused by unmeasured confounding. A sensitivity analysis places bounds on the allowable departure from a pair-randomized experiment. We use the model of \citet{ros87} and \citet[][\S 4]{obs}, which controls the allowable departure from a paired randomized experiment through a parameter $\Gamma  = \exp(\gamma) \geq 1$. In each pair, the model relates $u_{i1}$ and $u_{i2}$ to $\pi_i$ by $\pi_i = {\exp(\gamma u_{i1})}/\{\exp(\gamma u_{i1}) + \exp(\gamma u_{i2})\}$, which implies that $1/(1+\Gamma) \leq \pi_i\leq \Gamma/(1+\Gamma)$. The resulting model for biased treatment assignment in a paired observational study with sensitivity parameter $\Gamma$ is 
\begin{align}\label{eq:sensmodel}
\P(\bT = \bt\mid \cF, \cT) &= \prod_{i=1}^n\pi_i^{t_{i1}}(1-\pi_i)^{1-t_{i1}}, \;\; \frac{1}{1+\Gamma} \leq \pi_i \leq \frac{\Gamma}{1+\Gamma} \;\; (i=1,...,n).
\end{align}
$\Gamma=1$ recovers a paired randomized experiment, while $\Gamma> 1$ encodes a family of departures from unbiased assignments within each pair. Larger values of $\Gamma$ allow hidden bias to have a larger impact on the conditional assignment probabilities.

For ease of presentation, in what follows all expectations, variances, and probabilities will be implicitly computed conditional upon $\cF$ and $\cT$. For example, $\P(\bT = \bt \mid \cF, \cT)$ will be written as $\P(\bT=\bt)$ henceforth.

\subsection{Sensitivity analysis with constant effects}\label{sec:add}
Suppose interest lies in testing Fisher's sharp null hypothesis,
\begin{align*}H_{F}: \tau_{ij} = Y_{ij}(1) - Y_{ij}(0) = 0 \;\;\; \text{for all } i,j,\end{align*}
implying that the treatment has no effect for any of the $2n$ individuals in the study. Fisher's sharp null imputes the missing values of the potential outcomes, as $Y_{ij}^{obs} = Y_{ij}(1) = Y_{ij}(0)$ for all $i,j$ under $H_F$. Further, we see from (\ref{eq:pairdiff}) that $\hat{\tau}_i= (T_{i1}-T_{i2})(Y_{i1}^{obs} - Y_{i2}^{obs})= (T_{i1}-T_{i2})(Y_{i1}(0) - Y_{i2}(0))$ under $H_F$. As a result, $|\hat{\tau}_i|$ is fixed at $|Y_{i1}(0) - Y_{i2}(0)|$ across randomizations, with only the sign of $\hat{\tau}_i$ flipping according to the difference in treatment assignments $(T_{i1}-T_{i2})$.

Consider the average of the treated-minus-control paired differences, $\hat{\tau}= n^{-1}\sum_{i=1}^n\hat{\tau}_i$, commonly referred to in the context of randomization tests as the permutational $t$-statistic \citep{wel37, ros07}. As the missing potential outcomes are imputed under $H_F$, the observed value of $\hat{\tau}$ is computable for any feasible treatment assignment in the paired design. For any scalar $a$, the randomization distribution of $\hat{\tau}$ under $H_F$ is
\begin{align}\label{eq:null}
\P\left(\hat{\tau} \leq a \mid H_F\right) &= \sum_{\bt\in \Omega}\1\left\{n^{-1}\sum_{i=1}^n(t_{i1}-t_{i2})(Y_{i1}(0)-Y_{i2}(0))\leq a\right\}\P(\bT=\bt)\nonumber\\
&= \sum_{\bt\in \Omega}\1\left\{n^{-1}\sum_{i=1}^n(t_{i1}-t_{i2})(Y_{i1}(0)-Y_{i2}(0))\leq a\right\}\prod_{i=1}^n\pi_i^{t_{i1}}(1-\pi_{i})^{1-t_{i1}},\end{align} where $\1\{A\}$ is an indicator that the event $A$ occurred. In a paired experiment,  (\ref{eq:null}) reduces to the proportion of randomizations resulting in $\hat{\tau} \leq a$; however, in observational studies the indicators are weighted unequally according to $\P(\bT=\bt)$, which need not be uniform over $\Omega$ in the presence of hidden bias.

While the resulting value of $\hat{\tau}$ is known for each $\bt\in \Omega$ under $H_F$, the probability (\ref{eq:null}) remains unknown in a paired observational study as it depends on the unknown assignment probabilities. As a result, (\ref{eq:null}) cannot be directly employed as a reference distribution for testing Fisher's sharp null. A sensitivity analysis proceeds by, for a given value of $\Gamma$ in (\ref{eq:sensmodel}), finding the values for $\pi_i$ yielding the worst-case randomization distribution and worst-case $p$-value for the desired inference. One then increases the value of $\Gamma$ until the null hypothesis can no longer be rejected. This changepoint $\Gamma$ serves as a measure of the robustness of the study's findings to unmeasured confounding.

As an illustration, suppose we were to find the worst-case $p$-value corresponding for testing Fisher's sharp null subject to (\ref{eq:sensmodel}) holding at a particular value of $\Gamma$ with a greater-than alternative. Define $n$ random variables $V_{i, \Gamma}|\hat{\tau}_i|$, where $V_{i, \Gamma}$ are conditionally independent given $\cF$ and $\cT$ and take on the values $\pm$ 1 with
\begin{align*} \P(V_{i,\Gamma}=+1 ) &= \Gamma/(1+\Gamma),\\
\P(V_{i,\Gamma}=-1 ) &= 1/(1+\Gamma).\end{align*} For each $i$, observe that $V_{i,\Gamma}$  is constructed such that the largest possible probability is placed on +1 under the sensitivity model in (\ref{eq:sensmodel}), and hence such that $V_{i,\Gamma}|\hat{\tau}_i|$ is positive with the maximal probability. \citet{ros87,ros07} shows that under Fisher's sharp null, $\hat{\tau}_i$ is stochastically dominated by $V_{i, \Gamma}|\hat{\tau}_i|$, which has expectation $\{(\Gamma-1)/(1+\Gamma)\}|\hat{\tau}_i|$ under $H_F$.  That is, for any scalar $a$, $\P(\hat{\tau}_i\geq a\mid H_F) \leq \P(V_{i,\Gamma}|\hat{\tau}_i| \geq a \mid H_F)$ if (\ref{eq:sensmodel}) holds at $\Gamma$. 

Define $B_{i,\Gamma}$ to be centered versions of $V_{i,\Gamma}|\hat{\tau}_i|$ under Fisher's sharp null, 
\begin{align}\label{eq:bbar} 
B_{i,\Gamma} &= V_{i,\Gamma}|\hat{\tau}_i| - \left(\frac{\Gamma-1}{1+\Gamma}\right)|\hat{\tau}_i|,\end{align} such that $E(B_{i,\Gamma}) = 0$. Further define the random variable $D_{i,\Gamma}$ as the observed treated-minus-control paired difference minus the expectation of its bounding random variable,
\begin{align*}
 D_{i, \Gamma}&= \hat{\tau}_i - \left(\frac{\Gamma-1}{1+\Gamma}\right)|\hat{\tau}_i|.\end{align*} In comparing $B_{i,\Gamma}$ to $D_{i,\Gamma}$, note that while $B_{i,\Gamma}$ has expectation zero, $E(D_{i,\Gamma}\mid H_F) \leq 0$ if the sensitivity model (\ref{eq:sensmodel}) holds at $\Gamma$ as $\{(\Gamma-1)/(1+\Gamma)\}|\hat{\tau}_i|$ is the worst-case (largest) expectation for $\hat{\tau}_i$ under $H_F$.  As stochastic dominance is preserved under independent convolutions, the random variable $\bar{B}_\Gamma = n^{-1}\sum_{i=1}^nB_{i,\Gamma}$ stochastically bounds $\bar{D}_\Gamma = n^{-1}\sum_{i=1}^n D_{i,\Gamma}$ under Fisher's sharp null. Only when essential, we will write $\bar{B}_\Gamma = \bar{B}_\Gamma(\bV_\Gamma, \bhtau)$ to reflect its dependence on $V_{i,\Gamma}$ and $\hat{\tau}_i$.

Let $\hat{F}_\Gamma(a) = \P\left(\bar{B}_\Gamma \leq a \mid H_F\right)$ be the cumulative distribution function for the bounding random variable $\bar{B}_\Gamma$. To provide further insight into this bounding distribution, it is worthwile to consider how one would generate realizations from $\hat{F}_\Gamma(\cdot)$ by means of Monte Carlo simulation when conducting a sensitivity analysis with level of unmeasured confounding $\Gamma$ and a greater than alternative. Let $\bar{D}_\Gamma^{obs}$ be the observed value of $\bar{D}_\Gamma$, and consider the following simulation scheme.
\begin{algorithm}
\caption{Permutational $t$ sensitivity analysis for Fisher's sharp null at $\Gamma$}\label{alg:fisher}
\begin{enumerate} \item In the $m$th of $M$ iterations:
\begin{enumerate}
\item Generate $V_{i,\Gamma}\overset{iid}{\sim} 2\times Bernoulli\left(\frac{\Gamma}{1+\Gamma}\right)-1$ for each $i$.
\item Compute $B_{i,\Gamma} = V_{i,\Gamma}|\hat{\tau}_i| - \left(\frac{\Gamma-1}{1+\Gamma}\right)|\hat{\tau}_i|$ for each $i$.
\item Compute $\bar{B}^{(m)}_{\Gamma} = n^{-1}\sum_{i=1}^nB_{i,\Gamma}$; store this value across iterations.
\end{enumerate}
\item Approximate the bound on the greater-than $p$-value by\begin{align*} \hat{p}_{val} &= \frac{1 + \sum_{m=1}^M\1\{\bar{B}^{(m)}_{\Gamma} \geq \bar{D}_\Gamma^{obs}\}}{1+M}\end{align*}
\end{enumerate}
\end{algorithm}

To state the justification for this procedure more explicitly, if (\ref{eq:sensmodel}) holds at $\Gamma$, then for any scalar $a$ and any sample size $n$
\begin{align}
&\P\left(\bar{D}_\Gamma \geq a\mid H_F\right)\nonumber\\
&\leq \max_{p}\sum_{\bt\in \Omega}\1\left[n^{-1}\sum_{i=1}^n\left\{(t_{i1}-t_{i2})|\hat{\tau}_i| - \left(\frac{\Gamma-1}{1+\Gamma}\right)|\hat{\tau}_i|\right\} \geq a\right]\prod_{i=1}^np_i^{t_{i1}}(1-p_i)^{1-t_{i1}}\nonumber \\
&=\sum_{\bt\in \Omega}\1\left[n^{-1}\sum_{i=1}^n\left\{(t_{i1}-t_{i2})|\hat{\tau}_i| - \left(\frac{\Gamma-1}{1+\Gamma}\right)|\hat{\tau}_i|\right\} \geq a\right]\prod_{i=1}^n\left(\frac{\Gamma}{1+\Gamma}\right)^{t_{i1}}\left(\frac{1}{1+\Gamma}\right)^{1-t_{i1}}\nonumber \\
\label{eq:permt}& = \P\left(\bar{B}_\Gamma \geq a \mid H_F\right),
\end{align} such that $\bar{B}_\Gamma$ maximizes the right-tail probability subject to  (\ref{eq:sensmodel}) holding at $\Gamma$. Henceforth we will refer to the sensitivity analysis based on the randomization distribution of $\bar{B}_\Gamma$ as the permutational $t$ sensitivity analysis.

This procedure readily extends to a test of the null hypothesis that treatment effect is constant at some common value $\tau_0\neq 0$ for all individuals simply by replacing $\hat{\tau}_i$ with $\hat{\tau}_i - \tau_0$ in the preceding derivations, and extends to less-than alternatives by replacing $(\hat{\tau}_i - \tau_0)$ with $-(\hat{\tau}_i - \tau_0)$. That said, the procedure presented in this section does rely upon the assumption that the treatment effects are known for all individuals under the null hypothesis. For instance, if we merely assumed that the average of the treatment effects equaled zero while allowing for heterogeneous individual level effects, $\hat{\tau}_i \neq (T_{i1}-T_{i2})\left(Y_{i1}(0)-Y_{i2}(0)\right)$ in general. Ultimately, we will evaluate whether or not (\ref{eq:permt}) also bounds the maximal tail probability when testing whether or not the average of the $2n$ possibly heterogeneous treatment effects equals zero in the presence of hidden bias. Before doing so, we introduce a new method distinct from the permutational $t$ sensitivity analysis which does provide a valid sensitivity analysis while accommodating effect variation.

\section{Large-sample sensitivity analysis for the sample average treatment effect under effect heterogeneity}\label{sec:largesample}
\subsection{Neyman's notion of no effect}

The sample average treatment effect in a paired experiment or observational study, $\bar{\tau}$, is defined as the average of the treatment effects for the $2n$ individuals in our study, 
\begin{align*} \bar{\tau} = n^{-1}\sum_{i=1}^n\bar{\tau}_i = (2n)^{-1}\sum_{i=1}^n\sum_{j=1}^2\tau_{ij}.\end{align*} Forthcoming developments will focus on developing a valid level-$\alpha$ sensitivity analysis for the null hypothesis
\begin{align*} H_N: \bar{\tau} = 0,
\end{align*}
sometimes referred to as Neyman's weak null, when (\ref{eq:sensmodel}) is assumed to hold at a particular value for $\Gamma$.  If we further assume that the treatment effect is constant for all individuals, the null hypotheses $H_N$ and $H_F$ are equivalent and the sensitivity analysis described in \S \ref{sec:add} would be justified. In general however, the null hypothesis $H_N$ is {composite}, and there are infinitely many values for $\{\tau_{ij}$: $i=1,...,n$; $j=1,2\}$ satisfying $H_N$. The pattern of treatment effects specified by $H_F$ is simply one element of this composite null, which may or may not yield the worst-case $p$-value over all patterns of treatment effects allowed under $H_N$. Unlike tests for Fisher's sharp null, tests of $H_N$ in randomized experiments have historically been conducted using large-sample approximations rather than randomization tests. For instance, in paired experiments the conventional test of $H_N$ simply uses the parametric $t$-test based on the treated-minus-control difference in means \citep{ima08}. We now present a large-sample test for $H_N$ valid in paired observational studies with hidden bias.

\subsection{A Neyman-style sensitivity analysis}\label{sec:neyman}

Recall the definition of $D_{i,\Gamma}$ as
\begin{align*}
 D_{i, \Gamma}&= \hat{\tau}_i - \left(\frac{\Gamma-1}{1+\Gamma}\right)|\hat{\tau}_i|.
\end{align*} At $\Gamma=1$, note that $\bar{D}_1  = \hat{\tau}$, the average treated-minus-control paired difference. Further observe that while $|\hat{\tau}_i|$ is fixed at $|Y_{i1}(0)-Y_{i2}(0)|$ under $H_F$, $|\hat{\tau}_i|$ generally varies across randomizations in $\Omega$ for other elements of $H_N$, taking on values $|Y_{i1}(1)-Y_{i2}(0)|$ and $|Y_{i2}(1)-Y_{i1}(0)|$ with probability $\pi_i$ and $1-\pi_i$ respectively. Capturing the impact of $|\hat{\tau}_i|$ on the overall variation of $\bar{D}_\Gamma$ when $\Gamma>1$ is essential in what follows. Towards that end, let $se(\cdot)$ denote the conventional standard error estimator for the sample mean based upon $n$ observations, for example
\begin{align*}
\text{se}(\bar{D}_\Gamma)^2 &= \frac{1}{n(n-1)}\sum_{i=1}^n(D_{i, \Gamma} - \bar{D}_{\Gamma})^2.
\end{align*} At $\Gamma=1$, $\text{se}(\bar{D}_1)$ is the usual standard error in a paired experiment, as $\text{se}(\bar{D}_1)^2 = se(\hat{\tau})^2 = \{n(n-1)\}^{-1}\sum_{i=1}^n(\hat{\tau}_i - \hat{\tau})^2$. 

Fix $\alpha$ with $0<\alpha \leq 0.5$.  Consider a candidate level-$\alpha$ test that the sample average treatment effect equals $0$ with a greater-than alternative, and with allowable degree of bias controlled by $\Gamma$ in  (\ref{eq:sensmodel}) of the form\begin{align*}
\varphi_N(\alpha, \Gamma) &= \1\{\bar{D}_{\Gamma}/\text{se}(\bar{D}_\Gamma) \geq \Phi^{-1}(1-\alpha)\},
\end{align*} where $\Phi(\cdot)$ is the cumulative distribution function of the standard normal. $\varphi_N(\alpha,\Gamma)$ is simply the event that the candidate sensitivity analysis returns a rejection of the null hypothesis. At $\Gamma=1$ $\varphi_N(\alpha,1)$ is the conventional large-sample test for Neyman's weak null in a paired experiment, where one rejects the null hypothesis if $\hat{\tau}/se(\hat{\tau})$ exceeds the appropriate quantile from the standard normal. As the following theorem demonstrates, constructing a test through the random variables $\bar{D}_\Gamma$ allows for a natural extension of conventional tests for $H_N$ in paired experiments to a sensitivity analysis where $\Gamma>1$.


\begin{theorem}\label{prop:neyman}
Suppose that treatment assignment satisfies (\ref{eq:sensmodel}) for a particular $\Gamma \geq 1$. Under mild regularity conditions, 
\begin{align*}
\underset{n\rightarrow\infty}{\lim}E(\varphi_N(\alpha, \Gamma) \mid H_N) \leq \alpha,
\end{align*} such that the test $\varphi_N(\alpha,\Gamma)$ provides an asymptotically valid sensitivity analysis for testing Neyman's weak null hypothesis.
\end{theorem}

The regularity conditions, presented in the supplementary material, serve to preclude certain pathological sequences of potential outcomes to ensure, for instance, that a central limit theorem holds for $\bar{D}_{\Gamma}$. Under these conditions, Theorem \ref{prop:neyman} implies that for sufficiently large $n$, if we reject the null hypothesis $H_N$ when $\bar{D}_{\Gamma}/\text{se}(\bar{D}_\Gamma) \geq \Phi^{-1}(1-\alpha)$, then asymptotically we will incorrectly reject a true null hypothesis with probability at most $\alpha$ when (\ref{eq:sensmodel}) holds at $\Gamma$. That is, $\varphi_N(\alpha, \Gamma)$ provides an asymptotically valid level-$\alpha$ sensitivity analysis for Neyman's weak null.

The proof is divided into several lemmas, each of which illustrates an important component of the procedure. Those most essential to the result are presented in the appendix, while those stemming from standard derivations are deferred to the web-based supplementary material. Lemma \ref{lemma:ub} constructs a new random variable $\bar{U}_{\Gamma}$ that stochastically bounds $\bar{D}_{\Gamma}$ for any sample size $n$; however, its randomization distribution is not directly useful as it depends on the unknown values of the missing potential outcomes, which are not imputed by $H_N$. Lemma \ref{lemma:expec} shows that the conditional expectation of $\bar{U}_{\Gamma}$, and hence of $\bar{D}_{\Gamma}$, is bounded above by 0 when the sample average treatment effect equals 0 and (\ref{eq:sensmodel}) holds at $\Gamma$. Lemma \ref{lemma:var} illustrates that $\text{se}(\bar{D}_\Gamma)^2$ provides an estimator of $\var(\bar{D}_{\Gamma} )$ which is conservative in expectation regardless of the values for the unknown probabilities $\pi_i$. Lemma \ref{lemma:vub} shows that the expectation of $\text{se}(\bar{D}_\Gamma)^2$ is also larger than $\var(\bar{U}_{\Gamma} )$ when (\ref{eq:sensmodel}) holds at $\Gamma$. Together, these results bound, in expectation, the moments of the unknown stochastically dominating random variable $\bar{U}_{\Gamma}$ by quantities computable from the observational study at hand, and hold without any regularity conditions.  In the supplementary material, we illustrate that under suitable regularity conditions,  this sharp bounding random variable has a distribution which is asymptotically normal. We then demonstrate that despite the true moments for the bounding random variable being unknown, the true expectation of $\bar{U}_{\Gamma}$ can be safely replaced by zero, and the true variance of $\var(\bar{U}_{\Gamma} )$ similarly replaced by $se({\bar{D}_{\Gamma}})^2$ without corrupting the asymptotic size of the procedure.

\section{A studentized sensitivity analysis}\label{sec:studentwhole}
\subsection{An alternative reference distribution using biased randomizations}\label{sec:student}
While Theorem \ref{prop:neyman} provides a large-sample sensitivity analysis for Neyman's weak null, much of the elegance of the randomization-based sensitivity analysis for Fisher's sharp null has been lost along the way. Rather than constructing a biased randomization distribution to perform inference, inference by means of $\varphi_N(\alpha,\Gamma)$ simply compares the test statistic $\bar{D}_\Gamma/\text{se}(\bar{D}_\Gamma)$ to a critical value from a standard normal. Importantly, as simulations in \S \ref{sec:sim} illustrate, the performance of $\varphi_N$ can be poor in moderately sized samples, such that deriving a procedure with improved finite-sample performance is undoubtedly warranted. We now demonstrate that, with appropriate studentization, the worst-case randomization distribution developed when constructing the bounding random variable $\bar{B}_\Gamma$ defined in (\ref{eq:bbar}) can be employed towards this end.

 For any realization of $\bT$, define the biased randomization distribution $\hat{G}_{\Gamma}$ by
\begin{align}
\label{eq:Ghat}\hat{G}_{\Gamma}(x) &= \P\left(\frac{\bar{B}_\Gamma}{\text{se}(\bar{B}_\Gamma)}\leq x \mid \bT\right)\\&= \sum_{\bt \in \Omega}\1\left\{\frac{\bar{B}_{\Gamma}(\bt_1-\bt_2,\bhtau)}{\text{se}(\bar{B}_\Gamma(\bt_1-\bt_2,\bhtau))}\leq x\right\}\prod_{i=1}^n\left(\frac{\Gamma}{1+\Gamma}\right)^{t_{i1}}\left(\frac{1}{1+\Gamma}\right)^{1-t_{i1}}\nonumber,
\end{align}
where $\bt_j = (t_{1j}, t_{2j},...,t_{nj})$ contains the treatment assignments for the $j$th unit in each pair. Note that $\hat{G}_{\Gamma}(x)$ is itself a random variable: for each point $x$, it varies with the observed treatment $\bT$ by means of its dependence on the magnitude of the treated-minus-control paired differences $|\hat{\tau}_i|$ which vary across randomizations under $H_N$. Observe that $\hat{G}_{\Gamma}(x)$ utilizes the same biased distribution for treatment assignments as does (\ref{eq:permt}); however, it importantly computes the randomization distribution of the studentized statistic $\bar{B}_{\Gamma}/\text{se}(\bar{B}_\Gamma)$ instead of the unstudentized statistic $\bar{B}_\Gamma$.

To bring the procedure to life it is useful to consider how one can generate draws from $\hat{G}_\Gamma(\cdot)$ by means of Monte Carlo simulation when calculating a worst-case $p$-value with a greater than alternative at level of unmeasured confounding $\Gamma$. Let $\bar{D}^{obs}_\Gamma/\text{se}(\bar{D}^{obs}_\Gamma)$ be the observed value for $\bar{D}_\Gamma/\text{se}(\bar{D}_\Gamma)$ in the observational study at hand, and consider the following simulation scheme.
\begin{algorithm}
\caption{Studentized sensitivity analysis for Neyman's weak null at $\Gamma$}\label{alg:neyman}
\begin{enumerate} \item In the $m$th of $M$ iterations:
\begin{enumerate}
\item Generate $V_{i,\Gamma}\overset{iid}{\sim} 2\times Bernoulli\left(\frac{\Gamma}{1+\Gamma}\right)-1$ for each $i$.
\item Compute $B_{i,\Gamma} = V_{i,\Gamma}|\hat{\tau}_i| - \left(\frac{\Gamma-1}{1+\Gamma}\right)|\hat{\tau}_i|$ for each $i$.
\item Compute $\bar{B}_\Gamma = n^{-1}\sum_{i=1}^nB_{i,\Gamma}$.
\item Compute $\text{se}(\bar{B}_\Gamma)^2 = \{n(n-1)\}^{-1}\sum_{i=1}^n(B_{i,\Gamma} - \bar{B}_\Gamma)^2$.
\item Compute $S^{(m)}_\Gamma = \bar{B}_\Gamma/\text{se}(\bar{B}_\Gamma)$; store this value across iterations. 
\end{enumerate}
\item Approximate the bound on the greater-than $p$-value by \begin{align*} \hat{p}_{val} &= \frac{1 + \sum_{m=1}^M\1\{S^{(m)}_\Gamma\geq \bar{D}^{obs}_\Gamma/\text{se}(\bar{D}^{obs}_\Gamma)\}}{1+M}\end{align*}
\end{enumerate}
\end{algorithm}

Algorithms 1 and 2 use the same biased probabilities of assignment to treatment to generate $V_{i,\Gamma}$ in each iteration, and differ only in the studentization of $\bar{B}_\Gamma$ when forming the test statistic. Define a candidate level-$\alpha$ sensitivity analysis at $\Gamma$ for $H_N$ with a greater than alternative based on this studentized randomization distribution,
\begin{align*}
\varphi_S(\alpha, \Gamma) &= \1\left\{\bar{D}_{\Gamma}/\text{se}(\bar{D}_\Gamma)  \geq \hat{G}^{-1}_{ \Gamma}(1-\alpha)\right\},
\end{align*}  where $\hat{G}_\Gamma^{-1}(1-\alpha)= \inf\{x: \hat{G}_\Gamma(x)\geq 1-\alpha\}$ is the $1-\alpha$ quantile of the distribution $\hat{G}_\Gamma$. Comparing the new procedure to $\varphi_N$, $\varphi_S$ simply replaces a critical value from a normal approximation with one from the randomization distribution for $\bar{B}_\Gamma/\text{se}(\bar{B}_\Gamma)$. The following theorem justifies the use of $\hat{G}_\Gamma(\cdot)$ as a bounding distribution for the random variable $\bar{D}_\Gamma/\text{se}(\bar{D}_\Gamma)$: if the sensitivity model holds at $\Gamma$, then asymptotically conducting inference using $\hat{G}_\Gamma$ as a reference distribution controls the Type I error rate at $\alpha$ even with heterogeneous treatment effects.

\begin{theorem}\label{prop:student}
Under mild regularity conditions, for all $x$ and conditional upon $\cF$ and $\cT$,
\begin{align*}
\hat{G}_{{\Gamma}}(x) &\overset{p}{\rightarrow} \Phi(x),
\end{align*}
\end{theorem}

\begin{corollary}\label{cor:randdist}
Under mild regularity conditions, if treatment assignment satisfies (\ref{eq:sensmodel}) for a particular $\Gamma \geq 1$,
\begin{align*}\underset{n\rightarrow\infty}{\lim}E(\varphi_S(\alpha, \Gamma) \mid H_N) \leq \alpha\end{align*} when treatment assignment satisfies (\ref{eq:sensmodel}) at $\Gamma$ and Neyman's weak null holds.  \end{corollary}

The proof of Theorem \ref{prop:student} is presented in the supplementary web material, while Corollary \ref{cor:randdist} is an immediate consequence of Theorems \ref{prop:neyman}, \ref{prop:student} and Lemma 11.2.1 of \citet{leh05}.  To test the null that $\bar{\tau} = \bar{\tau}_0$ for general $\bar{\tau}_0$, one need simply replace $\hat{\tau}_i$ with $\hat{\tau}_i - \bar{\tau}_0$ in the definition of $D_{i,\Gamma}$, while a less-than alternative can be accommodated by replacing $\hat{\tau}_i - \bar{\tau}_0$ with $-(\hat{\tau}_i - \bar{\tau}_0)$. These results justify the use of the studentized randomization distribution $\hat{G}_{ \Gamma}$ as a null distribution for the test statistic $\bar{D}_{\Gamma}/\text{se}(\bar{D}_\Gamma)$ to conduct an asymptotically valid sensitivity analysis at level of unmeasured confounding $\Gamma$ even in the presence of effect heterogeneity. This facilitates evaluating the composite null $H_N$ with a single test which, by being based on a biased randomization distribution, maintains the nonparametric spirit of the many procedures for testing sharp nulls described in \citet{obs}.

\subsection{An exact and asymptotically robust sensitivity analysis}\label{sec:exact}
In the classical two-sample setup with samples of size $n$ and $m$ drawn $iid$ and independently from distributions $P$ and $Q$ respectively, it is well known that permutation tests, valid for the null that $P=Q$, need not provide Type I error control even asymptotically for the null $\theta(P) = \theta(Q)$ for real parameters $\theta$ such as the population mean. To address this, \citet{chu13} recommend studentization as a general mechanism for employing permutation-based reference distributions to furnish asymptotically valid inference for $\theta(P)=\theta(Q)$ while attractively maintaining finite-sample exactness if in reality $P=Q$. 

Parallels are readily drawn between tests of equality of distribution and tests of Fisher's sharp null, and likewise between tests of equality of expectations and tests of Neyman's weak null. Under no unmeasured confounding the procedure $\varphi_S(\alpha,1)$ is exact for Fisher's sharp null, while its asymptotic correctness under Neyman's weak null in a paired experiment follows from Theorem \ref{prop:student}. As currently constructed the procedure $\varphi_S(\alpha,\Gamma)$ is \textit{not} exact for $H_F$ under $\Gamma > 1$. This is due to $\bar{B}_\Gamma/\text{se}(\bar{B}_\Gamma)$ not being arrangement increasing function within pairs over $\Omega$ under Fisher's sharp null, a property essential for the theoretical development of sensitivity analyses assuming constant effects. As a result, $\bar{B}_\Gamma/\text{se}(\bar{B}_\Gamma)$ need not stochastically bound the distribution of $\bar{D}_\Gamma/\text{se}(\bar{D}_\Gamma)$ under Fisher's sharp null despite the fact that $\bar{B}_\Gamma$ \textit{does} bound the distribution of $\bar{D}_\Gamma$. See \citet[][\S 2.4.4]{obs} for an overview of arrangement increasing test statistics and their importance in sensitivity analysis. Fortunately the lack of stochastic ordering stems solely from behavior in the left tail of $\hat{G}_\Gamma(x)$ when testing a greater than alternative, while under this alternative our interest naturally lies in the behavior of the right tail. As a result, a slight modification to $\varphi_S$ simultaneously provides a sensitivity analysis that is finite-sample exact for $H_F$ and asymptotically valid for $H_N$.

Define the positive part randomization distribution $\hat{G}_{\Gamma+}$ by replacing $\bar{B}_\Gamma/\text{se}(\bar{B}_\Gamma)$ with \\$\max\{0, \bar{B}_\Gamma/\text{se}(\bar{B}_\Gamma)\}$ in (\ref{eq:Ghat}). The distribution can be equivalently expressed as
$\hat{G}_{\Gamma+}(x) = \hat{G}_\Gamma(x)\1\{x \geq 0\}$, from which it is seen that the positive part distribution piles the mass of the negative values in the support of $\hat{G}_\Gamma(x)$ at zero. Consider using the positive part distribution as a reference distribution for a sensitivity analysis,
\begin{align*}\varphi_{S+}(\alpha,\Gamma) &= \1\{\max\{0,\bar{D}_\Gamma/\text{se}(\bar{D}_\Gamma)\} \geq \hat{G}^{-1}_{\Gamma+}(1-\alpha)\}.
\end{align*} 
The following theorem demonstrates that the positive part modification $\varphi_{S+}$ provides an exact sensitivity analysis at $\Gamma$ for Fisher's sharp null, while achieving asymptotic Type I error control under Neyman's null.
\begin{theorem}\label{thm:exact}
Suppose that treatment assignment satisfies (\ref{eq:sensmodel}) for a particular $\Gamma \geq 1$. For any $n \geq 2$, under Fisher's sharp null,
\begin{align*} 
E(\varphi_{S+}(\alpha,\Gamma)\mid H_F)\leq \alpha.
\end{align*}
Furthermore, under Neyman's weak null and under mild regularity conditions,
\begin{align*} 
\underset{n\rightarrow\infty}{\lim}E(\varphi_{S+}(\alpha,\Gamma)\mid H_N)\leq \alpha.
\end{align*}
\end{theorem}
The proof of Theorem \ref{thm:exact}, deferred to the web-based supplement, requires showing that under $H_F$, the positive part test statistic is arrangement increasing. With this established, stochastic dominance follows directly from Theorem 2 of \citet{ros87}. 

Decisions based upon $\varphi_{S+}$ and $\varphi_{S}$ can only differ with large values for the desired significance level $\alpha$, specifically values of $\alpha$ such that critical value with a greater-than alternative would fall \textit{below} zero when using $\hat{G}_\Gamma(x)$ despite the fact that $E(\bar{B}_\Gamma)$ has expectation zero under the null. As $\hat{G}_\Gamma(x)$ converges in probability at each point $x$ to the CDF of a standard normal, a discrepancy between $\varphi_S$ and $\varphi_{S+}$ requires $\alpha \geq 0.5$ asymptotically, much larger than is convention. For commonly employed values of $\alpha$, the tests will coincide except in pathological instances of little practical concern. As a result, the original procedure $\varphi_S$ can be safely thought of as providing an exact and asymptotically robust sensitivity analysis, and our discussion henceforth will concern $\varphi_S$ without the positive part modification.

We call the resulting procedure the \textit{studentized sensitivity analysis}. The procedure maintains its exactness under Fisher's sharp null while providing asymptotic Type I error control even in the presence of heterogeneous effects. We now assess whether or not the studentization was necessary for testing $H_N$. That is, might the permutational $t$ sensitivity analysis using the unstudentized difference-in-means statistic also provide a valid sensitivity analysis for Neyman's weak null?


\section{The permutational $t$-test with hidden bias and effect heterogeneity}\label{sec:implication}
\subsection{Potential for improper size at $\Gamma > 1$}\label{sec:fail}
Recall that the random variable $\bar{B}_\Gamma$ stochastically dominates $\bar{D}_\Gamma$ under Fisher's sharp null as described in \S \ref{sec:add}. In this section, we compare the variance of the random variable $\bar{D}_\Gamma$ to that of $\bar{B}_\Gamma$ while allowing for heterogeneous effects. To do so, it is useful to define a new quantity representing the difference in the averages of the potential outcomes for the two individuals in a given pair, 
\begin{align*}
\eta_i = \frac{Y_{i1}(0) + Y_{i1}(1)}{2} - \frac{Y_{i2}(0) + Y_{i2}(1)}{2}.
\end{align*}
Using $\eta_i$, the treated-minus-control paired difference in any pair $i$ can be expressed as
\begin{align*}\hat{\tau}_i &= \bar{\tau}_i + (T_{i1}-T_{i2})\eta_i,
\end{align*}
such that the true variance of $\hat{\tau}_i$ given $\cF$ and $\cT$ is seen to depend on $\eta_i$ but not on $\bar{\tau}_i$. In the special case of Fisher's sharp null $\bar{\tau}_i = 0$ and $\eta_i = Y_{i1}(0)-Y_{i2}(0)$ for all pairs, recovering the setting of \S \ref{sec:add}.

$\bar{B}_{\Gamma}$, the random variable used to facilitate the permutational $t$ sensitivity analysis under Fisher's sharp null, has variance 
\begin{align}\label{eq:varB}\var(\bar{B}_{\Gamma}) &= \frac{4\Gamma}{n^2(1+\Gamma)^2}\left(\sum_{i=1}^n\bar{\tau}_i^2 + \eta_i^2 + 2\left(2\pi_i-1\right)\eta_i\bar{\tau}_i\right).
\end{align} The last term of (\ref{eq:varB}) would drop out if either Fisher's sharp null were true such that $\bar{\tau}_i=0$ for all $i$, or if $\pi_i=0.5$ as would be the case in a paired experiment. Meanwhile $\bar{D}_{\Gamma}$, the random variable whose distribution we seek to bound when testing $H_N$, has variance
\begin{align}\label{eq:varD}\var(\bar{D}_{\Gamma}) &= \frac{1}{n^2}\sum_{i=1}^n\pi_i(1-\pi_i)\left\{2\eta_i - \left(\frac{\Gamma-1}{1+\Gamma}\right)(|\bar{\tau}_i + \eta_i| - |\bar{\tau}_i - \eta_i|)\right\}^2,\end{align} which would further simplify to $4n^{-2}\sum_{i=1}^n\pi_i(1-\pi_i)\eta_i^2$ under Fisher's sharp null. 

Denote the permutational $t$ sensitivity analysis using the unstudentized difference-in-means described in \S \ref{sec:add} by
\begin{align*}
\varphi_F(\alpha,\Gamma)  = \1\{\bar{D}_\Gamma \geq \hat{F}_\Gamma^{-1}(1-\alpha)\},
\end{align*} where $\hat{F}_\Gamma(a) = \P(\bar{B}_\Gamma \leq a \mid H_F)$ is the worst-case distribution for $\bar{D}_\Gamma$ under the assumption of Fisher's sharp null and $\hat{F}_\Gamma^{-1}(1-\alpha)$ is its $1-\alpha$ quantile. If the sensitivity model (\ref{eq:sensmodel}) holds at $\Gamma=1$ and we correctly conduct inference at $\Gamma= 1$, then $\var(\bar{B}_{1}) = n^{-2}\sum_{i=1}^n(\bar{\tau}_i^2 + \eta_i^2)$, while under the same conditions $\var(\bar{D}_{1}) = n^{-2}\sum_{i=1}^n\eta_i^2 \leq \var(\bar{B}_{1})$. This underpins the well known result that the permutational $t$-test at $\Gamma=1$, $\varphi_F(\alpha,1)$, yields asymptotically conservative inference for $H_N$ in a paired experiment; see, for example, \citet{ima08} and \citet{din17}. Unfortunately, for $\Gamma > 1$ there exist allocations of potential outcomes yielding heterogeneous treatment effects satisfying Neyman's weak null for which $E(\bar{D}_\Gamma) = E(\bar{B}_\Gamma) = 0$ but where $\var(\bar{D}_{\Gamma}) > \var(\bar{B}_{\Gamma})$. This possibility, along with asymptotic normality of both $\bar{D}_{\Gamma}$ and $\bar{B}_{\Gamma}$, account for the following negative result concerning the permutational $t$ sensitivity analysis in the presence of effect heterogeneity.
\begin{theorem}\label{prop:fisher}
Consider conducting a sensitivity analysis at $\Gamma$ for $H_N$ using the permutational $t$, $\varphi_F(\alpha, \Gamma)$. Then, there exist sequences of potential outcomes satisfying the conditions of Theorem \ref{prop:student} and contained in $H_N$ such that, if (\ref{eq:sensmodel}) holds at $\Gamma$
\begin{align*}
\underset{n\rightarrow \infty}{\lim}E(\varphi_F(\alpha, \Gamma)  \mid H_N) > \alpha.
\end{align*}
That is, the permutational $t$ sensitivity analysis can fail to control the Type I error rate under effect heterogeneity if (\ref{eq:sensmodel}) holds at $\Gamma>1$ for certain patterns of effect heterogeneity in the composite null $H_N$.
\end{theorem}

The implication of Theorem \ref{prop:fisher} is that the permutational $t$ does not, in general, provide a valid level-$\alpha$ sensitivity analysis for Neyman's null, and consequently that the studentization developed in \S\S \ref{sec:largesample}-\ref{sec:studentwhole} was indeed well motivated. As a numerical illustration of Theorem \ref{prop:fisher}, let $n$ be even and consider the following allocation for $\bar{\tau}_i$, $\eta_i$, and $\pi_i$:  
\begin{align}\label{eq:counter}
\{\bar{\tau}_i, \eta_i, \pi_i\} &= \begin{cases} \{2.5,5,4/5\} & i=1,...,n/2\\
\{-2.5, 20, 4/5\} & i=n/2 + 1,...,n\end{cases}.\end{align}
While perhaps contrived, this represents a case where the permutational $t$ sensitivity analysis fails to control the Type I error rate even asymptotically. Here $\bar{\tau} = 0$ such that Neyman's weak null is true, (\ref{eq:sensmodel}) holds at $\Gamma=4$, but the effects are heterogeneous such that Fisher's sharp null is false. The random variable $\bar{D}_{4}$ has expectation  $0$ and variance $151.84/n$ from (\ref{eq:varD}), while $\bar{B}_{4}$ has expectation 0 and variance $125.6/n$ from (\ref{eq:varB}). $\bar{D}_4$ not only has an expectation equal to the worst case, but also has a larger variance than $\bar{B}_4$. We thus see that $\bar{D}_4$ has more mass in its right tail than $\bar{B}_4$ attributes, which results in an anti-conservative procedure as $\bar{B}_4$ is the random variable whose randomization distribution is employed in a sensitivity analysis based on the permutational $t$. Asymptotically, the permutational $t$ sensitivity analysis using the distribution of $\bar{B}_4$ as a reference distribution will reject while attempting to maintain the size at $\alpha$ if $n^{1/2}\bar{D}_{4} \geq (125.6)^{1/2}\Phi^{-1}(1-\alpha)$. Because $\var(\bar{D}_4) > \var(\bar{B}_4)$, this event actually occurs with probability $1-\Phi\{(125.6/151.84)^{1/2}\Phi^{-1}(1-\alpha)\} > \alpha$ if $\alpha < 0.5$. For example, at $\alpha = 0.05$, the permutational $t$ sensitivity analysis $\varphi_F(0.05, 4)$ has an asymptotic Type I error rate of 0.067, a reflection of Theorem \ref{prop:fisher}. 

\subsection{Proper asymptotic size at $\Gamma+\epsilon$}\label{sec:epsilon}
How severe is the potential anti-conservativeness of the permutational $t$ sensitivity analysis? A sensitivity analysis typically proceeds by iteratively increasing the value of $\Gamma$ being tested until we transition from rejecting the null hypothesis to failing to reject the null hypothesis. In large samples, the behavior of a sensitivity analysis is bias-dominated \citep{ros04}. As Lemma \ref{lemma:expec} in the appendix illustrates, $E(\bar{D}_\Gamma)\leq 0$ when (\ref{eq:sensmodel}) holds at $\Gamma$ and Neyman's weak null is true. Meanwhile, by centering within its construction $E(\bar{B}_\Gamma) = 0$. The permutational $t$ sensitivity analysis based on the unstudentized difference-in-means uses a randomization distribution that successfully bounds the expectation of $\bar{D}_\Gamma$, and the potential for size greater than $\alpha$ stems only from discrepancies in the variance in those instances when the upper bound on the expectation is tight. Under mild regularity conditions, $E(\bar{D}_\Gamma)$ is a decreasing a function of $\Gamma$ as described in the proof of Theorem \ref{prop:epsilon} in the supplementary material. By conducting a sensitivity analysis at $\Gamma+\epsilon$ when (\ref{eq:sensmodel}) holds at $\Gamma$, we create positive gap between $E(\bar{B}_{\Gamma+\epsilon})$ and $E(\bar{D}_{\Gamma+\epsilon})$. This gap persists asymptotically, while the standard errors scale at the usual rate.  There is hope, then, that the changepoint $\Gamma$ returned using the unstudentized difference-in-means may not be grossly unrepresentative of that of returned by the asymptotically valid studentized procedure $\varphi_S$. As we now formalize, asymptotically the changepoint $\Gamma$ for a sensitivity analysis conducted using the permutational $t$-statistic is arbitrarily close to that of the asymptotically valid procedures in the presence of effect heterogeneity. 

\begin{theorem}\label{prop:epsilon}
Suppose (\ref{eq:sensmodel}) holds at level $\Gamma$ and that the sample average treatment effect equals $\tau$. Consider conducting a sensitivity analysis at level of unmeasured confounding $\Gamma+\epsilon$ for any $\epsilon > 0$ by means of the permutational $t$ sensitivity analysis, $\varphi_F(\alpha, \Gamma+\epsilon)$. Then, under the assumptions of Theorem \ref{prop:fisher},

\begin{align*}\underset{n\rightarrow\infty}{\lim}\E(\varphi_F(\alpha, \Gamma+\epsilon)\mid H_N) &= 0\end{align*} 
That is, if (\ref{eq:sensmodel}) holds at $\Gamma$, the permutational $t$ sensitivity analysis asymptotically commits a Type I error with probability 0 when performed at $\Gamma+\epsilon$, despite potentially having size greater than $\alpha$ when the sensitivity analysis is conducted at $\Gamma$.
\end{theorem}
As a further illustration of Theorem \ref{prop:epsilon}, we return to the example given in in (\ref{eq:counter}). Suppose we conduct a sensitivity analysis using $\varphi_F(0.05, 4.01)$ despite that fact that (\ref{eq:sensmodel}) actually holds at $\Gamma=4$. The random variable $\bar{D}_{4.01}$ has expectation  $-0.01$ and variance $151.86/n$ from (\ref{eq:varD}), while $\bar{B}_{4.01}$ has expectation 0 and variance $125.42/n$ from (\ref{eq:varB}), such that $E(\bar{B}_{4.01}) > E(\bar{D}_{4.01})$ but $\var(\bar{B}_{4.01}) < \var(\bar{D}_{4.01})$. If we conduct a sensitivity analysis at $\Gamma=4.01$ the permutational $t$ sensitivity analysis rejects the null asymptotically if $\bar{D}_{4.01} \geq \Phi^{-1}(1-\alpha)(125.42/n)^{1/2}$, or equivalently if $(\bar{D}_{4.01}  + 0.01)/(151.86/n)^{1/2} \geq  \{\Phi^{-1}(1-\alpha)(125.42/n)^{1/2} + 0.01\}/(151.86/n)^{1/2}$. The left-hand side converges in distribution to a standard normal, while the right-hand side simplifies to $\{0.01 n^{1/2} + 125.42^{1/2}\Phi^{-1}(1-\alpha)\}/151.86^{1/2}$, which goes to $\infty$ as $n$ increases. Hence, the test rejects with probability 0 in the limit. In summary, while $\varphi_F(\alpha, 4)$ does not asymptotically control the Type I error rate for this example as illustrated in the previous section, $\varphi_F(\alpha, 4.01)$  does asymptotically. As Theorem \ref{prop:epsilon} indicates, Type I error control would be attained asymptotically by the procedure $\varphi_F(\alpha, 4+\epsilon)$ for any $\epsilon > 0$, i.e. simply by using the permutational $t$ sensitivity analysis at a slightly larger value of $\Gamma$ than necessary. Creating a positive gap between the actual expectation and the worst-case expectation swamps out any discrepancies in the variances asymptotically, another reflection of bias trumping variance in the presence of hidden bias.

\begin{table}
\begin{center}
\begin{tabular}{c|| c c c c c || c c c c c c}
&\multicolumn{5}{c||}{Permutational $t$: $\varphi_F(0.05, \Gamma)$} &\multicolumn{5}{c}{Studentized: $\varphi_S(0.05, \Gamma)$}\\
\diagbox{$n$}{$\Gamma$}& 4.00 & 4.05 &4.10 &4.20&4.40 & 4.00 & 4.05 & 4.10 &4.20&4.40\\
\hline\hline
50	&	0.076	&	0.074	&	0.072	&	0.059	&	0.048	&	0.061	&	0.059	&	0.052	&	0.041	&	0.032	\\
100	&	0.070	&	0.072	&	0.068	&	0.057	&	0.038	&	0.054	&	0.052	&	0.048	&	0.043	&	0.030\\
500	&	0.067	&	0.058	&	0.048	&	0.036	&	0.013	&	0.050	&	0.041	&	0.035	&	0.025	&	0.009	\\
1000	&	0.065	&	0.055	&	0.042	&	0.024	&	0.007	&	0.048	&	0.038	&	0.030	&	0.016	&	0.005	\\
5000	&	0.066	&	0.038	&	0.021	&	0.005	&	0.000	&	0.049	&	0.025	&	0.012	&	0.004	&	0.000	\\
\end{tabular}
\end{center}
\caption{\label{tab:finite} Size for the  permutational $t$ sensitivity analysis ($\varphi_F)$ and the studentized sensitivity analysis $(\varphi_S)$ for the example in (\ref{eq:counter}) with different sample sizes $n$ (rows) and at different values of $\Gamma$ (columns). The true minimal value of $\Gamma$ for which (\ref{eq:sensmodel}) holds is 4 in this example, such that all columns except for the first for each method conduct the sensitivity analysis at a larger value of $\Gamma$ than necessary, i.e. at $\Gamma=4+\epsilon$ for $\epsilon > 0$. For a sensitivity analysis to be asymptotically valid at a given $\Gamma$, the Type I error rate should fall at or below the desired level $\alpha=0.05$ as $n$ increases. }

\end{table}
The result of Theorem \ref{prop:epsilon} is asymptotic, and hence does not indicate how much larger $\Gamma$ may need to be to achieve Type I error control for finite samples through the permutational $t$. Table \ref{tab:finite} presents Monte Carlo estimates of the Type I error rates of both $\varphi_S(0.05, 4+\epsilon)$ and $\varphi_F(0.05, 4+\epsilon)$ using the example (\ref{eq:counter}) at $n=50, 100, 500, 1000, 5000$ and $\epsilon = 0, 0.05, 0.1, 0.2, 0.4$.  For each combination of $n$ and $\Gamma$, the Type I error rate was estimated based on 10,000 randomizations. For each randomization, the sensitivity analysis was conducted using 1000 Monte Carlo draws according to Algorithm 1 for $\varphi_F$ and Algorithm 2 for $\varphi_S$. The first column of each section of the table presents the sensitivity analysis using $\varphi_F$ and $\varphi_S$ at $\Gamma=4$, which is the minimial value of $\Gamma$ for which the model (\ref{eq:sensmodel}) holds in example (\ref{eq:counter}). Going down the rows for $\varphi_F$ at $\Gamma=4$ provides an illustration of Theorem \ref{prop:fisher}, as the permutational $t$ sensitivity analysis continues to be anti-conservative even at $n=5000$. Indeed, as calculated at the end of \S \ref{sec:fail}, the Type I error rate in this example converges to 0.067 in the limit, such that the anti-conservativeness persists asymptotically at $\Gamma=4$. Down the rows of $\varphi_S$ at $\Gamma=4$, we see the consequences of Theorem \ref{prop:student}, as the studentized sensitivity analysis does provide an asymptotically valid sensitivity analysis for Neyman's weak null. While not shown in the table, the tests $\varphi_{S+}$ and $\varphi_S$ gave the same $p$-value for each randomization and each value of $\Gamma$, further illustrating that the modification required to attain exactness under Fisher's null is of little practical importance or concern. The remaining columns of the table for $\varphi_F$ provide insight into Theorem \ref{prop:epsilon}. For the section corresponding to $\varphi_F$, we see that once the sensitivity analysis is conducted at $\Gamma > 4$, the Type I error rate falls below 0.05 for sufficiently large sample sizes, and indeed will converge to zero in the limit. How large $n$ needs to be to result in the appropriate size in finite samples depends upon $\Gamma$: as the value of $\Gamma$ employed becomes more conservative, the required sample size decreases. For $\varphi_S$ the Type I error rate also goes to zero at $\Gamma > 4$, and in the limit the primary differences between the two methods will lie in their Type I error rates at $\Gamma=4$. 

It is important to keep in mind that the results of this simulation study are not representative of the  typical reality in a paired observational study: the example (\ref{eq:counter}) was actively chosen to illustrate Theorem \ref{prop:fisher}. In many observational studies, the inference conducted at the smallest $\Gamma$ such that (\ref{eq:sensmodel}) holds is itself conservative, as the inference uses the worst-case probability assignments for any $\Gamma$. If all pairs are not affected in the worst possible way, the worst-case expectation employed by the permutational $t$ will be larger than the true expectation. In such observational studies, $\varphi_F$ will typically itself be conservative even at $\Gamma$; see \citet{fog18extend} for additional discussion along with a method accounting for heterogeneous degrees of unmeasured confounding.

\subsection{The benefits of the studentized sensitivity analysis} \label{sec:sim}
We now illustrate through a simulation study that the large-sample sensitivity analysis $\varphi_N$ employing critical values based on the normal distribution can itself be considerably anti-conservative in small samples. In so doing, we highlight the ability of the studentized sensitivity analysis to capture departures from normality in small samples, leading us to recommend the studentized procedure over its large-sample approximation. We once again proceed with the allocation of $\{\bar{\tau}_i, \eta_i, \pi_i\}$ in (\ref{eq:counter}), and set $n=100$. We then generate 10,000 realizations from the resulting biased randomization distribution (\ref{eq:sensmodel}) with this allocation of biased probabilities. Within each realization, we test the null hypothesis $\bar{\tau} = 0$ against the alternative that $\bar{\tau} > 0$ with desired level $\alpha=0.05$ and allowable bias $\Gamma=4$ using $\varphi_N(0.05, 4)$, $\varphi_S( 0.05, 4)$, and $\varphi_F(0.05, 4)$. For $\varphi_S$ and $\varphi_F$, we replace $\hat{F}^{-1}_{4}(0.95)$ and $\hat{G}^{-1}_{4}(0.95)$ with Monte Carlo estimates based on 10,000 randomizations.

\begin{figure}
\begin{center}
\includegraphics[scale=.7]{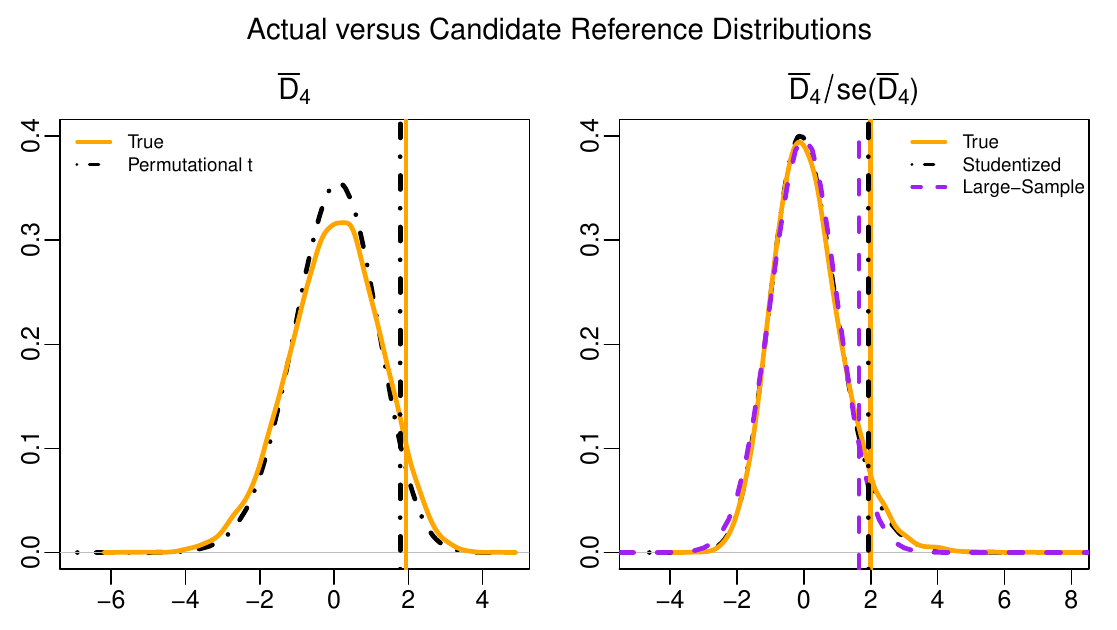}
\caption{Bounding the randomization distribution in the presence of hidden bias in the simulation study of \S \ref{sec:sim}. The sample average treatment effect in this simulation equals 0, and the smallest $\Gamma$ for which (\ref{eq:sensmodel}) holds equals 4. (Left) The left-hand side shows the true distribution of $\bar{D}_{4}$ with a solid line, while the dotted line shows the attempted  bounding distribution utilized by the unstudentized difference-in-means. (Right) The right-hand side shows the true distribution of $\bar{D}_{4}/\text{se}(\bar{D}_4)$ (solid), along with the bounding distributions from the studentized sensitivity analysis (dotted) and from the large-sample normal approximation (dashed). In both figures, the vertical lines correspond to the 0.95 quantile of the displayed distributions.}\label{fig:dist}
\end{center}
\end{figure}

Figure \ref{fig:dist} shows the true distributions of $\bar{D}_{4}$ (left) and $\bar{D}_{4}/\text{se}(\bar{D}_4)$ (right), along with the distributions utilized by the candidate sensitivity analyses at $\Gamma=4$. For the sensitivity analysis to control the size at $\alpha$, the bounding distribution needs a fatter right tail than the true distribution, such that the quantiles generated by each procedure should fall to the right of those of the true distribution. The left-hand side of Figure \ref{fig:dist} illustrates that for the permutational $t$ sensitivity analysis the opposite holds: the quantiles of the true distribution fall to the right of those of the candidate worst-case distribution, a reflection of Theorem \ref{prop:fisher}. The estimated Type I error rate for the unstudentized difference-in-means based sensitivity analysis exceeded 0.05, and was in fact 0.0702. That is, the permutational $t$ sensitivity analysis fails to control the Type I error rate. The right-hand side of Figure \ref{fig:dist} illustrates that the Neyman-style large-sample test, while valid in sufficiently large samples by Theorem \ref{prop:neyman}, also fails to bound the true distribution of $\bar{D}_{4}/\text{se}(\bar{D}_{4})$ in this finite sample simulation. The large-sample test simply uses a normal approximation to the distribution of $\bar{D}_{4}/\text{se}(\bar{D}_{4})$, while the figure illustrates that the true distribution exhibits skewness. The Type I error rate is even worse than that of the permutational $t$ sensitivity analysis, estimated at 0.0798. The studentized sensitivity analysis, asymptotically valid by Theorem \ref{prop:student}, is able to capture the skewness in the true distribution of $\bar{D}_{4}/\text{se}(\bar{D}_4)$. Figure \ref{fig:dist} shows that the estimated $95$th percentile for the studentized test is virtually identical to that of the true distribution. This yields a Type I error rate of 0.054, coming much closer to the desired level $\alpha=0.05$.

\section{Data examples and illustrations}\label{sec:example}
\subsection{Four paired observational studies}
We now compare sensitivity analyses conducted using the permutational $t$-test and the studentized sensitivity analysis in four paired observational studies. Through doing so, we further highlight the practical similarities in reported insensitivity to hidden bias attained through these two methods despite the differences in theoretical guarantees presented in \S\S \ref{sec:studentwhole}-\ref{sec:implication}. We now briefly describe the four observational studies to be analyzed.
\begin{itemize}
\item \textbf{Alcoholism and Genetic Damage (n=20)}. \citet{maf00} paired 20 alcoholics (defined as having consumed $>120$ grams of per alcohol per day) with 20 non-alcoholics (consuming between 8 and 13 grams per day) on the basis of covariates such as gender, age, and smoking habit to assess the impact of excessive alcohol intake on genetic damage.
\item \textbf{Welding and Genetic Damage (n=39)}. \citet{wer98} paired $n=39$ male welders to male non-welders using smoking habits and age to assess the impact of occupational exposure to potential carginogens such as nickel and cadmium on genetic damage.
\item \textbf{Smoking and Lead Levels (n=250)}. \citet{ros13} paired $n=250$ daily smokers to non-smokers controlling for gender, age, race, education level, and household income, and compared blood lead levels between the two groups. The data were from the 2007-2008 National Health and Nutrition Examination Survey (NHANES).
\item  \textbf{Smoking and Periodontal Disease (n=441)}. \citet{ros16} paired $n=441$ daily smokers to non-smokers controlling for gender, age, race, education level, and household income. The outcome variable was a measure of disease on the 14 lower teeth for each individual, with larger values indicating a more severe case of periodontal disease. The data were from the 2011-2012 National Health and Nutrition Examination Survey (NHANES).

\end{itemize}

These observational studies vary considerably in terms of sample size, ranging from $n=20$ to $n=441$, hence representing a comparison of the two methods in small and moderate sample regimes. In each study, arguments can be made for the treatment effect varying across individuals. How misleading might the sensitivity analysis based on the permutational $t$-test be relative to an analysis using the studentized sensitivity analysis, which is asymptotically valid for Neyman's weak null accommodating heterogeneous effects?

\subsection{Sensitivity values and intervals}\label{sec:result}
For each observational study, sensitivity analyses were conducted using both the permutational $t$-test and the studentized sensitivity analysis at $\alpha=0.01$, $0.05$, and $0.10$. For each sensitivity analysis, we found the largest value of $\Gamma$ such that the sensitivity analysis continued to reject the null. This changepoint value of $\Gamma$ is also known as the \textit{sensitivity value} of an observational study \citep{zha18}, and quantifies the magnitude of unmeasured confounding required to alter the findings of the observational study. It would be disconcerting if $\varphi_F$ and $\varphi_S$ could return drastically different sensitivity values for a given observational study; fortunately, the results summarized in Table \ref{tab:change} reveal that this was not the case for any of the sensitivity analyses conducted. Indeed, the sensitivity value returned by the permutational $t$ never exceeded that of the studentized approach by more than a factor of 1.01. In some instances, $\varphi_F$ actually returned a lower reported insensitivity to unmeasured confounding than $\varphi_S$. Investigating precisely for which distributions of treated-minus-control paired differences can be expected to occur is an ongoing area of research.

\begin{table}
\begin{center}
\begin{tabular}{c || c c | c c | c c | c c }
\multicolumn{1}{c}{}&\multicolumn{2}{c}{\citet{maf00}}& \multicolumn{2}{c}{\citet{wer98}}&\multicolumn{2}{c}{\citet{ros13}}& \multicolumn{2}{c}{\citet{ros16}}\\
\multicolumn{1}{c}{}&\multicolumn{2}{c}{$n=20$}& \multicolumn{2}{c}{$n=39$}& \multicolumn{2}{c}{$n=250$}&\multicolumn{2}{c}{$n=441$}\\

$\alpha$ & Perm. $t$ & Student &Perm. $t$&Student &Perm. $t$&Student&Perm. $t$&Student\\
\hline
0.01&3.132&3.138 & 3.029 & 2.994&1.640 & 1.628&2.392&2.433 \\
0.05&5.041 & 5.042& 4.231 & 4.239 &1.908 & 1.901&2.657 & 2.701\\
0.10&6.655 & 6.655& 5.150 & 5.208 &2.078 & 2.073&2.817& 2.856\\
\end{tabular}
\caption{\label{tab:change} The largest value of $\Gamma$ such that the null hypothesis is rejected using both the permutational $t$ and the studentized sensitivity analyses in four observational studies. Results are shown for $\alpha=0.01, 0.05$ and $0.10$.}

\end{center}
\end{table}

As further demonstration of the similarities between the two methods, we now compare sensitivity intervals constructed using the two tests. A sensitivity interval is an extension of a confidence interval to observational studies with hidden bias bounded by a particular $\Gamma$. Its interpretation remains the same: a 100(1-$\alpha$) sensitivity interval is the set of treatment effects $\bar{\tau}$ cannot be rejected by a two-sided level$-\alpha$ sensitivity analysis conducted at $\Gamma$, and can thus be attained through the inversion of a test. Those attained using the permutational $t$ implicitly assume the effect is constant, while those attained through the studentized sensitivity analysis make no such assumption. At $\Gamma=1$, the studentized sensitivity analysis furnishes confidence intervals which are asymptotically equivalent to the conventional large-sample confidence intervals in a paired experiment, of the form $\hat{\tau} \pm \Phi^{-1}(1-\alpha/2)\text{se}(\hat{\tau})$. 

Table \ref{tab:SI} compares 90\% sensitivity intervals attained using $\varphi_F$ and $\varphi_S$ for these four paired observational studies at $\Gamma=3$. Once again, while the methods do not result in identical intervals they are quite similar in terms of both length and values at the endpoints. In two studies the studentized sensitivity analysis resulted in larger intervals, while in the other two the permutational $t$ provided larger intervals. Importantly, systematic differences do not exist between the sensitivity intervals created through the permutational $t$ and the studentized sensitivity analysis.

\begin{table}
\begin{center}
\begin{tabular}{c|| c  c | c c}
 \multicolumn{1}{c}{}& \multicolumn{2}{c}{Permutational $t$, $\Gamma=3$} & \multicolumn{2}{c}{Studentized, $\Gamma=3$} \\
Study &90\% Interval & Length & 90\% Interval & Length\\
\hline
\citet{maf00}&(1.13, 9.63) & 8.50 & (1.18, 9.42)&8.24\\
\citet{wer98}&(0.11, 1.14) &1.03& (0.11, 1.15)&1.04\\
\citet{ros13} &(-0.33, 1.99) & 2.32& (-0.35, 2.20)&2.55\\
\citet{ros16} &(-0.70, 15.88)&16.58& (-0.61, 15.92)&16.53
\end{tabular}
\caption{\label{tab:SI} 90\% sensitivity intervals at $\Gamma=3$ through inverting the permutational $t$ sensitivity analysis and the studentized sensitivity analysis in the four observational studies.}

\end{center}
\end{table}

%
\
\section{Concluding remarks}
The subtleties of the constant treatment effects model, along with the differences in the implications of additivity in randomized experiments versus observational studies, may be lost on practitioners employing sensitivity analyses. Researchers may well have in mind the null of no effect on average when performing a sensitivity analysis using the unstudentized difference-in-means, justified under the assumption of constant effects. Applied sensitivity analyses typically report the minimal value of $\Gamma$ for which, at a given level $\alpha$, the hypothesis test fails to reject the null. The sensitivity analysis based on the unstudentized difference-in-means, $\varphi_F(\alpha, \Gamma)$, fails to control the Type I error rate when (\ref{eq:sensmodel}) holds at $\Gamma$ in the presence of effect heterogeneity; however as Theorem \ref{prop:epsilon} demonstrates, simply conducting a permutational $t$ sensitivity analysis at $\Gamma+\epsilon$ for any $\epsilon > 0$ eliminates the problem asymptotically. From a practical perspective, this suggests that the reported changepoint values of $\Gamma$ for which studies can no longer reject the null hypothesis when using the traditional approach, and with it the perceived robustness of a study's findings to unmeasured confounding, will likely not be substantively larger than what is justified. The discrepancies, or lack thereof, in the observational studies analyzed in \S \ref{sec:result} align with this narrative.

Based on the perception at the time, \citet{cop97} wrote of Rosenbaum's model that it ``can be applied to various generalizations of the signed rank test...but would be more complicated for other statistics (e.g. a $t$-test)'' \citep[][p. 71]{cop97}. The studentized sensitivity analysis not only overcomes these difficulties (indeed its large-sample equivalent presented in \S \ref{sec:largesample} directly extends the parametric $t$-test to observational studies), but also demonstrates that the effort was worth the while. The studentized sensitivity analysis provides an exact sensitivity analysis for Fisher's sharp null while maintaining the desired level asymptotically under Neyman's weak null by using a biased randomization distribution along with studentization to conduct inference. The method presented herein provides a natural modification of the permutational $t$ sensitivity analysis for accommodating heterogeneous effects. Under constant effects, the permutational $t$ typically performs worse in a sensitivity analysis than other choices of test statistics, such as the signed rank test, certain $u$-statistics, and certain $m$-statistics \citep{ros07, ros11}, leading it to not be favored for sensitivity analyses in practice. The extent to which studentized versions of those tests provide valid sensitivity analyses for the sample average treatment effect, or for other meaningful causal estimands, remains an open question. The machinery underpinning the results presented herein should prove beneficial in addressing that question.

This work serves to further dispel the notion that the model of \citet[][\S 4]{obs} is only useful for testing sharp null hypotheses. Sensitivity analyses for composite nulls on binary outcomes have been developed in \citet{ros02att} and \citet{fog17}, while in the case of continuous outcomes \citet{ros03} presents an exact sensitivity analysis for the Walsh averages. By providing a sensitivity analysis for the sample average treatment effect while accommodating effect heterogeneity, we hope to further enable and encourage researchers to conduct sensitivity analyses when inferring treatment effects in observational studies.

\appendix
\section{Proofs}
\subsection{Notation, regularity conditions, and a note on asymptotics}

Before presenting the regularity conditions and additional notation to facilitate the proofs, a clarification regarding the theorems in the main article is warranted. Theorems 1-5 are presented conditional upon Neyman's null $H_N$ being true, such that $n^{-1}\sum_{i=1}^n\bar{\tau}_i = 0$; however, for Neyman's null to be true at each point along an asymptotic sequence it would have to be the case that $\bar{\tau}_i = 0$ for all pairs. For precision, the asymptotics should instead reflect a sequence of observational studies of increasing sample size. For each number of pairs $n$ the treatment effects in each pair $\bar{\tau}_i$ should be adaptively re-centered by the sample average treatment effect in the first $n$ pairs, call it $\bar{\tau}^{(n)}$, such that Neyman's null would hold after the re-centering. For instance, for each $n$ the random variable $D_{i,\Gamma}$ should be re-defined as
\begin{align*}
D_{i,\Gamma}^{(n)} = \hat{\tau}_i - \bar{\tau}^{(n)} - (2\theta_\Gamma-1)|\hat{\tau}_i - \bar{\tau}^{(n)}|
\end{align*}
when considering the results in Theorems 1-5. We have omitted this in the text and in the proofs in this web supplement, trading precision for notational convenience and enhanced readability. 

The following regularity conditions are imposed throughout the proofs that follows.  As a reminder, the quantity $\eta_i$ is defined as in the main article to represent the difference in the averages of the potential outcomes for the two individuals in a given pair, 

\begin{condition}\label{cond:1}There exist constants $C > 0$, $\mu_m$ and $\mu_a$ such that as $n\rightarrow\infty$
\begin{align}\label{eq:CLT}
n^{-1}\sum_{i=1}^n|\eta_i| > C,\;\;  n^{-1}\sum_{i=1}^n\eta_i^2 > C,\\
\label{eq:cheby} n^{-2}\sum_{i=1}^n\eta_i^2\rightarrow 0,\;\; n^{-2}\sum_{i=1}^n\eta_i^4 \rightarrow 0,\;\; n^{-2}\sum_{i=1}^n\bar{\tau}_i^4 \rightarrow 0,\\
\label{eq:1moment}n^{-1}\sum_{i=1}(2\pi_i-1)\eta_i \rightarrow \mu_m,\;\;\; n^{-1}\sum_{i=1}^n\pi_i|\bar{\tau}_i+\eta_i| + (1-\pi_i)|\bar{\tau}_i- \eta_i| \rightarrow \mu_a.
\end{align}
\end{condition}
\begin{condition}\label{cond:2}
There exists a constant $\nu^2 > 0$ such that
\begin{align} \label{eq:2moment} n^{-1}\sum_{i=1}^n\pi_i(\bar{\tau}_i+ \eta_i)^2 + (1-\pi_i)(\bar{\tau}_i - \eta_i)^2\rightarrow \nu^2. \end{align}
\end{condition}

In the proofs that follow, let $\theta_\Gamma = \Gamma/(1+\Gamma)$, such that if the sensitivity model holds at $\Gamma$ we have that $1-\theta_\Gamma \leq \pi_i \leq \theta_\Gamma$ and that $(2\theta_\Gamma-1) = (\Gamma-1)/(1+\Gamma)$. Further, all results should be viewed as conditional upon $\cF$ and $\cT$; this has been omitted for enhanced readability.

\subsection{Theorem 1}
Define $n$ new random variables $U_{i,\Gamma}$ by 
\begin{align*} U_{i,{\Gamma}} &= \bar{\tau}_i + V_{i,\Gamma}|\eta_i| - \left(2\theta_\Gamma-1\right)\{(1 + V_{i,\Gamma})|\bar{\tau}_i + |\eta_i|| + (1-V_{i, \Gamma})|\bar{\tau}_i - |\eta_i||\}/2,\end{align*} where the random variables $V_{i,\Gamma}$ are distributed as in \S \ref{sec:add}.
\begin{lemma}\label{lemma:ub}
Suppose treatment assignment satisfies (\ref{eq:sensmodel}) at $\Gamma$.  Then, for any scalar $k$,
\begin{align}\label{eq:ub} \P(\bar{D}_{ \Gamma} \geq k) &\leq \P\left(\bar{U}_{\Gamma}\geq k\right) 
\end{align}Moreover, the upper bound is sharp in that sense that (\ref{eq:ub}) holds if $\pi_i=\theta_{\Gamma}$ and $\eta_i \geq -\eta_i$ for $i=1,...,n$.
\end{lemma}
\begin{proof}
$(T_{i1}-T_{i2})\eta_i = \pm |\eta_i|$ and, by (\ref{eq:sensmodel}), $1-\theta_{\Gamma} \leq \P((T_{i1}-T_{i2})\eta_i = |\eta_i| ) \leq \theta_{\Gamma}$. Now, $V_{i,\Gamma}|\eta_i| = \pm |\eta_i|$, and $\P(V_{i,\Gamma}|\eta_i| = |\eta_i| ) = \theta_\Gamma$, so that $V_{i,\Gamma}|\eta_i|$ stochastically dominates $(T_{i1}-T_{i2}) \eta_i$. As the function $x - (2\theta_\Gamma-1)|x|$ is monotone nondecreasing for all $x$ and for any $\Gamma\geq 1$, we have that $\bar{\tau}_i + |\eta_i| - (2\theta_\Gamma-1)|\bar{\tau}_i + |\eta_i|| \geq \bar{\tau}_i - |\eta_i| - (2\theta_\Gamma-1)|\bar{\tau}_i - |\eta_i||$. The random variable $D_{i, \Gamma}$ is thus stochastically dominated by $U_{i, \Gamma}$ for each $i$, and (\ref{eq:ub}) then follows from a standard probability inequality \citep[see, e.g.,][Lemma 3.3]{ahm81}.
\end{proof}
\begin{lemma}\label{lemma:expec} If (\ref{eq:sensmodel}) holds at $\Gamma$, then
\begin{align*}
E(\bar{U}_{\Gamma}) \leq 4\theta_\Gamma(1-\theta_\Gamma)\bar{\tau}.
\end{align*}
In particular, if $H_N$ is true, then $E(\bar{U}_{\Gamma}) \leq 0$.
\begin{proof}
 For each $i$, $E(U_{i, \Gamma}) = (2\theta_\Gamma-1)|\eta_i| + \bar{\tau}_i  - (2\theta_\Gamma-1)\{\theta_\Gamma|\bar{\tau}_i + |\eta_i|| + (1-\theta_\Gamma)|\bar{\tau}_i  - |\eta_i||\}$.  We now show that $\theta_\Gamma|\bar{\tau}_i + |\eta_i|| + (1-\theta_\Gamma)|\bar{\tau}_i| \geq |\eta_i| + (2\theta_\Gamma-1)\bar{\tau}_i$. We do this in three cases depending upon the values for $sign(\bar{\tau}_i + |\eta_i|)$ and $sign(\bar{\tau}_i - |\eta_i|)$

\begin{case}[$\bar{\tau}_i + |\eta_i| \geq 0$, $\bar{\tau}_i - |\eta_i| \geq 0$]

Here $\bar{\tau}_i \geq |\eta_i|$. Recalling that $0 \leq 2\theta_\Gamma-1 \leq 1$, \begin{align*}
\theta_\Gamma|\bar{\tau}_i + |\eta_i|| + (1-\theta_\Gamma)|\bar{\tau}_i - |\eta_i|| &= (2\theta_\Gamma-1)|\eta_i| + \bar{\tau}_i\\ &\geq |\eta_i| + (2\theta_\Gamma-1)\bar{\tau}_i.\end{align*}\end{case}
\begin{case}[$\bar{\tau}_i + |\eta_i| \geq 0$, $\bar{\tau}_i - |\eta_i|  < 0$]

Here we have that the result holds with equality, as $\theta_\Gamma|\bar{\tau}_i + |\eta_i|| + (1-\theta_\Gamma)|\bar{\tau}_i - |\eta_i|| =  |\eta_i| + (2\theta_\Gamma-1)\bar{\tau}_i$. \end{case}

\begin{case}[$\bar{\tau}_i + |\eta_i| < 0$, $\bar{\tau}_i - |\eta_i|  < 0$]

In this case $-(\bar{\tau}_i - \tau) \geq |\eta_i|$. Noting $-\bar{\tau}_i = - 2\theta_\Gamma\bar{\tau}_i + (2\theta_\Gamma-1)\bar{\tau}_i$ and that $2\theta_\Gamma, (2\theta_\Gamma-1)\geq 0$,\begin{align*}
\theta_\Gamma|\bar{\tau}_i + |\eta_i|| + (1-\theta_\Gamma)|\bar{\tau}_i - |\eta_i|| &= (1-2\theta_\Gamma)|\eta_i| - \bar{\tau}_i\\ &\geq |\eta_i| + (2\theta_\Gamma-1)\bar{\tau}_i.\end{align*} \end{case}
The inequality thus always holds. Hence,
\begin{align*}E(U_{i, \Gamma}) 
&\leq \left(2\theta_\Gamma-1\right)|\eta_i| + \bar{\tau}_i - \left(2\theta_\Gamma-1\right)\left\{|\eta_i| + (2\theta_\Gamma - 1)\bar{\tau}_i\right\}\\
&= 4\theta_\Gamma(1-\theta_\Gamma)\bar{\tau}_i.\end{align*} As this holds for all $i$, the result for the average follows.
\end{proof}
\end{lemma}
\begin{lemma}\label{lemma:var}
For any constant $\Gamma\geq 1$ 
\begin{align*}
\E(\text{se}(\bar{D}_\Gamma)^2 )- \var(\bar{D}_{ \Gamma}) &= \frac{1}{n(n-1)} \sum_{i=1}^n\left(E(D_{i,  \Gamma} ) - E(\bar{D}_{\Gamma} )\right)^2 \geq 0.
\end{align*}
\begin{proof}

\begin{align*}
\E(\text{se}(\bar{D}_\Gamma)^2)&=\frac{1}{n(n-1)}\left\{\sum_{i=1}^n E(D_{i,  \Gamma}^2) - n^{-1}\sum_{k,\ell=1}^n\E(D_{k,   \Gamma}D_{\ell, \Gamma})\right\}\\
&= n^{-2}\left\{\sum_{i=1}^n\E(D_{i,  \Gamma}^2) - \frac{1}{n-1}\sum_{k\neq \ell}E(D_{k,  \Gamma})E(D_{\ell,  \Gamma}\right\}\\
&= n^{-2}\left[\sum_{i=1}^n\left\{\var(D_{i,  \Gamma}) + E(D_{i,  \Gamma})^2\right\} - \frac{1}{n-1}\sum_{k\neq \ell}E(D_{k,  \Gamma})E(D_{\ell,  \Gamma})\right]\\
&= \var(\bar{D}_{ \Gamma}) + \frac{1}{n(n-1)}\sum_{i=1}^n\left\{E(D_{i,  \Gamma}) - E(\bar{D}_{\Gamma} )\right\}^2,
\end{align*}
proving the result.
\end{proof}
\end{lemma}
\begin{remark}\label{remark:var}The result of Lemma \ref{lemma:var} applies beyond the collection of random variables $\{D_{i, \Gamma}\}$. Take any collection of $n$ independent random variables $\{X_i\}$ with $E(X_i) = \mu_{i}$ and $\var(X_i) = \sigma^2_i$, and consider their random average $\bar{X}$. Then, $E(se(\bar{X})^2) - \var(\bar{X}) = ((n-1)n)^{-1}\sum_{i=1}^n(\mu_{i} - \bar{\mu})^2 \geq 0$.
\end{remark}

\begin{lemma}\label{lemma:vub} 
If (\ref{eq:sensmodel}) holds at $\Gamma$, then
\begin{align*}
\var\left(\bar{U}_{ \Gamma}\right) &\leq \E(\text{se}(\bar{D}_\Gamma)^2 ) \end{align*}
\begin{proof}
 By Lemma \ref{lemma:var}, it suffices to show that $\var(D_{i,  \Gamma}) \geq \var(U_{i, \Gamma})$ for all $i$. Since (\ref{eq:sensmodel}) holds at $\Gamma$, \begin{align*}
 \var(D_{i,  \Gamma} ) &= \pi_i(1-\pi_i)\left\{2\eta_i - \left(2\theta_\Gamma-1\right)(|\bar{\tau}_i + \eta_i| - |\bar{\tau}_i - \eta_i|)\right\}^2\\ &\geq \theta_{\Gamma}(1-\theta_{\Gamma})\left\{2|\eta_i| - \left(2\theta_\Gamma-1\right)(|\bar{\tau}_i + |\eta_i|| - |\bar{\tau}_i - |\eta_i||)\right\}^2 = \var(U_{i,\Gamma}),\end{align*} proving the result. 
\end{proof}
\end{lemma}
\begin{lemma}\label{lemma:momentbounds} 
For each $i$, 
\begin{align}\label{eq:U2moment}
16\theta_{\Gamma}(1-\theta_{\Gamma})^3\eta_i^2 \leq \var(U_{i,  \Gamma} ) \leq 16\theta_{\Gamma}^3(1-\theta_{\Gamma})\eta_i^2\\ 
\label{eq:U4moment}E(U_{i,  \Gamma}^4) \leq 128\theta_{\Gamma}^4\left(\bar{\tau}_i^4 + \eta_i^4\right)
\end{align}
Further, if treatment assignment satisfies (\ref{eq:sensmodel}) at $\Gamma$,
\begin{align}\label{eq:D2moment}
16\theta_{\Gamma}(1-\theta_{\Gamma})^3\eta_i^2 \leq \var(D_{i,  \Gamma} ) \leq 4\theta_{\Gamma}^2\eta_i^2\\ 
\label{eq:D4moment}E(D_{i,  \Gamma}^4) \leq 128\theta_{\Gamma}^4\left(\bar{\tau}_i^4 + \eta_i^4\right)
\end{align}
\begin{proof}
To prove (\ref{eq:U2moment}), observe that $\var(U_{i,\Gamma} ) = \theta_{\Gamma}(1-\theta_{\Gamma})(2\eta_i - (2\theta_{\Gamma}-1)(|\bar{\tau}_i+ \eta_i| - |\bar{\tau}_i-\eta_i|))^2$, which is at least $16\theta_{\Gamma}(1-\theta_{\Gamma})^3\eta_i^2$ and at most $16\theta_{\Gamma}^3(1-\theta_{\Gamma})\eta_i^2$. The proof of (\ref{eq:D2moment}) simply replaces $\theta_{\Gamma}(1-\theta_{\Gamma})$ with $1/4$ in the upper bound.

Proving (\ref{eq:U4moment}) requires multiple applications of $(a+b)^2 \leq 2a^2+2b^2$ for scalars $a$ and $b$. Without loss of generality assume that $\eta_i \geq -\eta_i$.
\begin{align*}E(U_{i,\Gamma}^4 ) &= \theta_{\Gamma}(\bar{\tau}_i + \eta_i - (2\theta_{\Gamma}- 1)|\bar{\tau}_i+\eta_i|)^4\\& + (1-\theta_{\Gamma})(\bar{\tau}_i- \eta_i - (2\theta_{\Gamma}-1)|\bar{\tau}_i-\eta_i|)^4\\
& \leq \theta_{\Gamma}(2\theta_{\Gamma}(\bar{\tau}_i+\eta_i))^4 + (1-\theta_{\Gamma})(2\theta_{\Gamma}(\bar{\tau}_i-\eta_i))^4\\
&\leq 128\theta_{\Gamma}^4\left(\bar{\tau}_i^4 + \eta_i^4\right).
\end{align*}
The proof of (\ref{eq:D4moment}) is analogous.
\end{proof}
\end{lemma}
\begin{lemma}\label{lemma:kstar} Both $n^{1/2}\bar{D}_{\Gamma}$ and $n^{1/2}\bar{U}_{\Gamma}$ are asymptotically normal. Further, let \begin{align}\label{eq:kstar2}k_{\Gamma}^*(\alpha) =  \Phi^{-1}(1-\alpha)\left\{n^{-2}\sum_{i=1}^n\var(U_{i,\Gamma} )\right\}^{1/2}.\end{align}
Then, if (\ref{eq:sensmodel}) holds at $\Gamma$ and $H_N$ is true,
\begin{align}\label{eq:kstar} \underset{n\rightarrow\infty}{\lim}\P\{\bar{D}_{{\Gamma}} \geq k^*_{\Gamma}(\alpha) \} \leq \alpha.
\end{align}
\end{lemma}
\begin{proof}
We prove asymptotic normality of $n^{1/2}\bar{U}_{\Gamma}$, and with it (\ref{eq:kstar}) by reference to Lemma \ref{lemma:ub}; the proof for $n^{1/2}\bar{D}_{\Gamma}$ is analogous. The $U_{i,\Gamma}$ are conditionally independent given $\cF$ and $\cT$. Further, by Lemma \ref{lemma:expec} and sharpness of $U_{i,\Gamma}$ as a stochastic upper bound we have that $E\left(n^{-1}\sum_{i=1}^nU_{i,\Gamma} \right) \leq 0$. To prove asymptotic normality of $n^{1/2}\bar{U}_\Gamma$, it suffices to show that Lyapunov's condition holds for $\delta=2$, i.e. that 
\begin{align*}
\sum_{i=1}^nE|U_{i,\Gamma}-E(U_{i,\Gamma})|^4/\left(\sum_{i=1}^n\var(U_{i,\Gamma})\right)^2 \rightarrow 0
\end{align*}
By (\ref{eq:U2moment}), $n^{-1}\sum_{i=1}^n\var(U_{i,\Gamma}) \geq 16\theta_{\Gamma}(1-\theta_{\Gamma})^3n^{-1}\sum_{i=1}^n\eta_i^2$, which is greater than $16\theta_{\Gamma}(1-\theta_{\Gamma})^3C$ for some $C > 0$ as $n\rightarrow\infty$ by (\ref{eq:CLT}). Applying Jensen's inequality and utilizing (\ref{eq:U4moment}) and (\ref{eq:cheby}), we have that $n^{-2}\sum_{i=1}^nE|U_{i,\Gamma}-E(U_{i,\Gamma})|^4\rightarrow 0$. Hence,
\begin{align*}
\sum_{i=1}^nE|U_{i,\Gamma}-E(U_{i,\Gamma})|^4/\left(\sum_{i=1}^n\var(U_{i,\Gamma})\right)^2 &= n^{-2}\sum_{i=1}^nE|U_{i,\Gamma}-E(U_{i,\Gamma})|^4/\left(n^{-1}\sum_{i=1}^n\var(U_{i,\Gamma})\right)^2\\ &\leq n^{-2}\sum_{i=1}^nE|U_{i,\Gamma}-E(U_{i,\Gamma})|^4/(16\theta_{\Gamma}(1-\theta_{\Gamma})^3C)^2 \rightarrow 0.
\end{align*}
This, along with Lemma \ref{lemma:ub}, proves the result.
\end{proof}

\begin{lemma}\label{lemma:prob}
Suppose that treatment assignment satisfies (\ref{eq:sensmodel}) at $\Gamma$ and $H_N$ holds. If  (\ref{eq:CLT}) and (\ref{eq:cheby}) hold, then for all $\epsilon > 0$, as $n\rightarrow \infty$

\begin{align}\label{eq:pmean}
\P\left(-\epsilon + \bar{D}_{\Gamma} \geq 0 \right) \rightarrow 0. \\
\label{eq:pvar}
\P\left\{\epsilon + n\text{se}(\bar{D}_\Gamma)^2 \leq n^{-1}\sum_{i=1}^n\var(U_{i,\Gamma})\right\} \rightarrow 0
\end{align}
\end{lemma}
\begin{proof}
We begin by proving (\ref{eq:pmean}). By Lemma \ref{lemma:expec},
$\P\left(-\epsilon + \bar{D}_\Gamma \geq 0\right) \leq \P\left(-\epsilon + \bar{D}_\Gamma-E(\bar{D}_\Gamma) \geq 0\right)$. The variance of $\var(D_{i,\Gamma})$ is, by (\ref{eq:D2moment}), less than $4\theta_{\Gamma}^2\eta_i^2$. Therefore, using (\ref{eq:cheby}),
\begin{align*}
\var\left(\bar{D}_\Gamma\right) & \leq 4\theta_{\Gamma}^2n^{-2}\sum_{i=1}^n\eta_i^2\rightarrow 0
\end{align*} as $n\rightarrow\infty$. Chebyshev's inequality then yields (\ref{eq:pmean}).  

We now prove (\ref{eq:pvar}). Recall that $n\text{se}(\bar{D}_\Gamma)^2 = (n-1)^{-1}\sum_{i=1}^nD_{i,\Gamma}^2 - n/(n-1)(\bar{D}_\Gamma)^2$. By Lemma \ref{lemma:vub}, 
\begin{align*} &\P\left\{\epsilon + n\text{se}(\bar{D}_\Gamma)^2 \leq n^{-1}\sum_{i=1}^n\var(U_{i})\right\} \\&\leq \P\left\{\epsilon + n\text{se}(\bar{D}_\Gamma)^2 \leq (n-1)^{-1}\sum_{i=1}^n\var(D_{i,\Gamma}) + (n-1)^{-1}\sum_{i=1}^n(E(D_{i,\Gamma}) - E(\bar{D}_\Gamma))^2\right\}\\®
&= \P\left\{\epsilon + (n-1)^{-1}\left(\sum_{i=1}^nD_{i,\Gamma}^2-\sum_{i=1}^nE(D_{i,\Gamma}^2)\right) - n(n-1)^{-1}(\bar{D}_\Gamma^2 - E(\bar{D}_\Gamma)^2)\leq 0\right\}\\
\end{align*} 

The proof of (\ref{eq:pmean}) along with (\ref{eq:1moment}) yields that $\bar{D}_\Gamma^2 - E(\bar{D}_\Gamma)^2$ converges in probability to 0. We now show that $(n-1)^{-1}\left\{\sum_{i=1}^nD_{i,\Gamma}^2- \sum_{i=1}^nE(D_{i,\Gamma}^2)\right\}$ also converges in probability to 0. Using (\ref{eq:D4moment}),

\begin{align*}
\var \left\{(n-1)^{-1}\sum_{i=1}^nD_{i,\Gamma}^2 \right\} &\leq (n-1)^{-2} \sum_{i=1}^nE(D_{i,\Gamma}^4 )\\
&\leq 128\theta_{\Gamma}^4(n-1)^{-2}\sum_{i=1}^n(\bar{\tau}_i^4 + \eta_i^4),
\end{align*} which converges to 0 as $n\rightarrow \infty$ through (\ref{eq:cheby}). Applying Chebyshev's inequality yields the desired convergence in probability, which in turn yields (\ref{eq:pvar}).
\end{proof}

\subsubsection*{Proof of Theorem \ref{prop:neyman}}
Define $k_{\Gamma}(\alpha) = \text{se}(\bar{D}_\Gamma)\Phi^{-1}(1-\alpha)$ with $0 < \alpha \leq 0.5$. By (\ref{eq:pvar}), taking $\epsilon \downarrow 0$,
\begin{align*}
\underset{n\rightarrow \infty}{\lim}\P\{k_{\Gamma}(\alpha) \geq  k^*_{\Gamma}(\alpha)\} = 1.
\end{align*}
This, in combination with (\ref{eq:kstar}), yields the conclusion of the theorem.
\subsection{Theorem \ref{prop:student}}
\begin{lemma}\label{lemma:bivariate}
Take a vector $\bV_\Gamma$ distributed as in \S \ref{sec:add} with $V_{i,\Gamma} = \pm 1$ and $\P(V_{i,\Gamma}  = 1) = \theta_\Gamma$. Let $\bV'_{\Gamma}$ be an $iid$ copy of $\bV_\Gamma$. Then, under (\ref{eq:cheby}) and (\ref{eq:2moment}), $n^{1/2}\bar{B}_{\Gamma}(\bV_\Gamma, \bhtau)$ and $n^{1/2}\bar{B}_{\Gamma}(\bV'_{\Gamma}, \bhtau)$ are $iid$ and converge jointly to a bivariate normal, each with mean zero and variance $ 4\theta_\Gamma(1-\theta_\Gamma)\nu^2$.
\begin{proof}
Recall that $B_{i,\Gamma} = V_{i,\Gamma}|\hat{\tau}_i| - (2\theta_\Gamma-1)|\hat{\tau}_i|$ and that $E(B_{i,\Gamma} ) = 0$. Since uncorrelatedness implies independence for the normal, to show independence of the limiting distributions for $\bar{B}_\Gamma$ and $\bar{B}'_\Gamma$  it suffices to show that $\text{cov}(\bar{B}_\Gamma, \bar{B}'_\Gamma ) = 0$.
\begin{align*}
\text{cov}(\bar{B}_\Gamma, \bar{B}'_\Gamma) &= E\left\{\text{cov}(\bar{B}_\Gamma, \bar{B}'_\Gamma \mid \bhtau) \right\}\\& + \text{cov}\left\{E(\bar{B}_\Gamma \mid  \bhtau ),E(\bar{B}'_\Gamma \mid  \bhtau) \right\}\\
&= n^{-2}E\left\{\sum_{i=1}^n\hat{\tau}_i^2\text{cov}(V_{i,\Gamma}, V'_{i,\Gamma}) \right\} + 0\\
&= 0
\end{align*}
By the Cram\'er-Wold device, to show bivariate asymptotic normality it suffices to show that $n^{1/2}(w_1\bar{B}_\Gamma + w_2\bar{B}'_\Gamma)$ converge to a normal with mean zero and variance $(w_1^2+w_2^2)4\theta_\Gamma(1-\theta_\Gamma)\nu^2$ for any vector of constants $(w_1, w_2)$. Fixing $(w_1, w_2)$, we now show this to be the case through Lyapunov's condition. We have $E(w_1B_{i,\Gamma} + w_2B_{i,\Gamma}) = 0$, that $E(B_{i,\Gamma}^4) = \pi_i(\bar{\tau}_i + \eta_i)^4 + (1-\pi_i)(\bar{\tau}_i - \eta_i)^4 \leq 8\bar{\tau}_i^4 + 8\eta_i^4$, and that $E((w_1B_{i,\Gamma} + w_2A'_{i,\Gamma})^4) \leq 8(w_1^4 + w_2^4)E(B_{i,\Gamma}^4)$.  Combining this with (\ref{eq:cheby}), we have that $n^{-2}\sum_{i=1}^nE((w_1B_{i,\Gamma} + w_2B_{i,\Gamma}')^4) \rightarrow 0$. By (\ref{eq:2moment}), we have that $n^{-1}\sum_{i=1}^n\var(w_1B_{i,\Gamma} + w_2B_{i,\Gamma}' ) \rightarrow (w_1^2+w_2^2)4\theta_\Gamma(1-\theta_\Gamma)\nu^2 > 0$. Hence,
\begin{align*}
&\sum_{i=1}^nE((w_1B_{i,\Gamma} + w_2B_{i,\Gamma}')^4 )/\left(\sum_{i=1}^n\var(w_1B_{i,\Gamma} + w_2B_{i,\Gamma}')\right)^2\\ &= n^{-2}\sum_{i=1}^nE((w_1B_{i,\Gamma} + w_2B_{i,\Gamma}')^4)/\left(n^{-1}(w_1^2+w_2^2)\sum_{i=1}^n\var(B_{i,\Gamma})\right)^2\rightarrow 0. 
\end{align*}
Lyapunov's condition is satisfied at $\delta=2$, proving the result.
\end{proof}
\end{lemma}

\begin{lemma}\label{lemma:F}
Under the assumptions of Theorem \ref{prop:student}, for any point $a$
\begin{align*}
\hat{F}_{ \Gamma}(a/n^{1/2}) \overset{p}{\rightarrow} \Phi\left(a/\nu_\Gamma\right),
\end{align*}
where $\nu^2_\Gamma = 4\theta_\Gamma(1-\theta_\Gamma)\nu^2$
\begin{proof}
Observe that 
\begin{align*} E(\hat{F}_{ \Gamma}(a/n^{1/2}) ) &= E(E(\1\{n^{1/2}\bar{B}_{\Gamma}(\bV_\Gamma, \bhtau)\leq a\} \mid  )  )\\
&= E(\1\{n^{1/2}\bar{B}_{\Gamma}(\bV_\Gamma, \bhtau)\leq a\} )\\
&= \P(n^{1/2}\bar{B}_{\Gamma}(\bV_\Gamma, \bhtau)\leq a) \end{align*}By Lemma \ref{lemma:bivariate}, $n^{1/2}\bar{B}_{\Gamma}(\bV_\Gamma, \bhtau)$ converges in distribution to a normal with mean 0 and variance $4\theta_\Gamma(1-\theta_\Gamma)\nu^2$. Hence, $E\{\hat{F}_{\Gamma}(a/n^{1/2})\} \rightarrow \Phi(a/\nu_\Gamma)$. Through Chebyshev's inequality, to illustrate the desired convergence in probability it suffices to show that $E\{\hat{F}_{\Gamma}^2(a/n^{1/2})\}\rightarrow \Phi^2(a/\nu_{\Gamma})$, which is equivalent to $\var\{\hat{F}_{\Gamma}(a/n^{1/2}) \} \rightarrow 0$.

\begin{align*}&E(\hat{F}^2_{\Gamma}(a/n^{1/2}))\\&= E\left[\sum_{\bt, \bt' \in \Omega}\1\{n^{1/2}\bar{B}_{\Gamma}(\bt_1-\bt_2,\bhtau) \leq a\}\1\{n^{1/2}\bar{B}_{\Gamma}(\bt_1'-\bt_2', \bhtau)\leq a\}\prod_{i=1}^n\theta_{\Gamma}^{t_{i1}+t_{i1}'}(1-\theta_{\Gamma})^{2-t_{i1}-t_{i1}'} \right]\\
&= \P(n^{1/2}\bar{B}_{\Gamma}(\bV_\Gamma, \bhtau) \leq a, n^{1/2}\bar{B}_{\Gamma}(\bV'_\Gamma, \bhtau) \leq a) \rightarrow \Phi^2(a/\nu_\Gamma)
\end{align*}as desired, where the last line uses Lemma \ref{lemma:bivariate}.
\end{proof}
\end{lemma}
\begin{lemma}\label{lemma:SA}
\begin{align*}
E\{\text{se}(\bar{B}_\Gamma)^2\} &= \var(\bar{B}_{ \Gamma} )
\end{align*}
\begin{proof} The lemma follows by Remark \ref{remark:var} along with the fact that by construction $B_{i,\Gamma}$ is centered, such that $E(\bar{B}_{i,\Gamma}) = 0$.
\end{proof}
\end{lemma}
\begin{lemma}\label{lemma:cprob}Under the assumptions of Theorem \ref{prop:student}, 
\begin{align*}
n\text{se}(\bar{B}_\Gamma)^2& \overset{p}{\rightarrow} 4\theta_\Gamma(1-\theta_\Gamma)\nu^2
\end{align*}
\begin{proof}
 Decompose $n\text{se}(\bar{B}_\Gamma)^2 = (n-1)^{-1}\sum_{i=1}^nB_{i,\Gamma}^2 - n/(n-1)\bar{B}_\Gamma$.  $E(\bar{B}_\Gamma)$ is 0, while by Lemma \ref{lemma:bivariate} $\bar{B}_\Gamma$ has limiting variance $4\theta_\Gamma(1-\theta_\Gamma)\nu^2/n \rightarrow 0$. Hence, $n/(n-1)\bar{B}_\Gamma$ converges in probability to 0 by Chebyshev's inequality. Meanwhile, $E((n-1)^{-1}\sum_{i=1}^nB_{i,\Gamma}^2 ) \rightarrow 4\theta_\Gamma(1-\theta_\Gamma)\nu^2$ by Lemma \ref{lemma:SA} and (\ref{eq:2moment}). To show that the variance of this term goes to zero, observe that $B_{i,\Gamma}^4 \leq 8(1+(2\theta_\Gamma-1)^4)\hat{\tau}_i^4$. Similar arguments to those of Lemma \ref{lemma:prob}, utilizing (\ref{eq:cheby}), then yield that the variance goes to zero, thus yielding the result through Chebyshev's inequality.
\end{proof}
\end{lemma}

\subsubsection*{Proof of Theorem \ref{prop:student}}
By Lemma \ref{lemma:F}, $\hat{F}_{\Gamma}(a/n^{1/2})$, the biased randomization distribution of $n^{1/2}\bar{B}_{\Gamma}$, converges in probability to $\Phi(a/\nu_\Gamma)$ for all points $a$, where again $\nu^2_\Gamma = 4\theta_\Gamma(1-\theta_\Gamma)\nu^2$. By Lemma \ref{lemma:cprob} and the continuous mapping theorem $n^{1/2} \text{se}(\bar{B}_\Gamma)$ converges in probability to $\nu_\Gamma$. Recall that $\hat{G}_{\Gamma}(t)$ is the biased randomization distribution of the studentized statistic $n^{1/2}\bar{B}_{\Gamma}/\{n^{1/2}\text{se}(\bar{B}_\Gamma)\}$. Setting $a=tn^{1/2}\text{se}(\bar{B}_\Gamma)$ and using Slutsky's theorem for randomization distributions \citep[][Lemma 5.2]{chu13}, we have that $\hat{G}_{\Gamma}(t)$ then converges in probability to $\Phi(t\nu_\Gamma/\nu_\Gamma) = \Phi(t)$ for all points $t$ as desired.

\subsection{Proof of Theorem 3}
We first prove exactness of $\varphi_{S+}(\alpha,\Gamma)$ under $H_F$. Re-arrange the pairs such that the first individual in each pair has the larger response. Define $q_{i1} = 2(1-\theta_\Gamma)|\eta_i|$ and $q_{i2} = -2\theta_\Gamma|\eta_i|$, and recall that $\theta_\Gamma \geq 1-\theta_\Gamma$. For any treatment assignment $\bt$, the positive part statistic can be expressed as 
\begin{align}\label{eq:teststat}
f(\bt, \bq) &= \max\left\{0,\frac{n^{-1}\sum_{i=1}^n\sum_{j=1}^2q_{ij}t_{ij}}{\sqrt{\frac{1}{n(n-1)}\sum_{i=1}^n\sum_{j=1}^2t_{ij}\left(q_{ij} - n^{-1}\sum_{i=1}^n\sum_{j=1}^2q_{ij}t_{ij}\right)^2}}\right\}
\end{align}
Let the vector $\bq[i12]$ equal the vector obtained from $\bq$ by exchanging the first and second elements in pair $i$ while leaving the other elements fixed, such that $(q[i12])_{i2} = q_{i1}$. A function $h(\bt,\bq)$ is called an arrangement increasing function in pairs if for all pairs $i$ $h(\bt,\bq) \geq h(\bt,\bq[iå12])$ whenever $(t_{i1}-t_{i2})(q_{i1}-q_{i2})\geq 0$ \citep[][\S 2.4.4]{obs}. In words, this says the function $h$ takes on a larger value when the elements $\bt$ and $\bq$ are arranged in the same order within a pair than it does when they are out of order. 

We now show that $f(\bt, \bq)$ in (\ref{eq:teststat}) is arrangement increasing in pairs. We do so for the $nth$ pair without loss of generality. For each $\bt \in \Omega$, let $d_i = t_{i1}q_{i1} + t_{i2}q_{i2}$. Consider fixed values for $d_1,...,d_{n-1}$ and consider the two possibilites for $d_n$, either $d_n=2(1-\theta_\Gamma)|\eta_n|$ or $d_n = -2\theta_\Gamma|\eta_n|$. It suffices to show that the function $f(\bt, \bq)$ is at least as large when $d_n = 2(1-\theta_\Gamma)|\eta_n|$ as it is when $d_{n} = -2\theta_\Gamma|\eta_n|$ for any fixed values of $d_1,...,d_{n-1}$. 

If $\sum_{i=1}^{n-1}d_i \leq 0$, then this is trivially true, as the test statistic will either be positive when $d_n=2(1-\theta_\Gamma)|\eta_n|$ and zero otherwise, or will be zero in both cases due to the positive part modification. We thus restrict attention the case $\sum_{i=1}^{n-1}d_i \geq 0$. As the numerator of $f(\bt, \bq)$ would be larger when $d_n = 2(1-\theta_\Gamma)|\eta_n|$, it is enough to show that the denominator will be smaller when $d_n = 2(1-\theta_\Gamma)|\eta_n|$ than it would be if $d_{n} = -2\theta_\Gamma|\eta_n|$ when $\sum_{i=1}^{n-1}d_i\geq 0$. Algebra yields that this is true if and only if $4(n-1)/n(1-\theta_\Gamma)^2\eta_n^2 - 4/n(1-\theta)|\eta_i|\sum_{i=1}^{n-1}d_i \leq 4(n-1)/n(\theta_\Gamma)^2\eta_n^2 + 4/n(1-\theta)|\eta_i|\sum_{i=1}^{n-1}d_i$, which holds as $\sum_{i=1}^{n-1}d_i \geq 0$ and $\theta_\Gamma \geq (1-\theta_\Gamma)$. 

The function $f(\bt, \bq)$ is thus arrangement increasing over randomizations in $\Omega$, and the proof of exactness of $\varphi_{S+}$ under Fisher's sharp null follows by applying Theorem 2 of \citet{ros87}. The asymptotic correctness of $\varphi_{S+}$ under Neyman's weak null follows in a straightforward way from Theorem 2, and the proof is omitted.

\subsection{Proof of Theorem \ref{prop:epsilon}}
Note that $\theta_{\Gamma+\epsilon} > \theta_\Gamma$ for any $\epsilon>0$. Consider $D_{i, \Gamma+\epsilon} = \hat{\tau}_i - (2\theta_{\Gamma+\epsilon}-1)|\hat{\tau}_i|$. Since (\ref{eq:sensmodel}) holds at $\Gamma$ under the Theorem's conditions, by arguments parallel to those in Lemma \ref{lemma:ub} $\bar{D}_{\Gamma+\epsilon}$ is stochastically bounded by the random variable $\bar{W}_{\Gamma,\epsilon}$, where 
\begin{align*} W_{i, \Gamma, \epsilon} &= \bar{\tau}_i + V_{i,\Gamma}|\eta_i| - (2\theta_{\Gamma+\epsilon}-1)\{(1 + V_{i,\Gamma})|\bar{\tau}_i + |\eta_i|| + (1-V_{i, \Gamma})|\bar{\tau}_i - |\eta_i||\}/2.\end{align*}
Hence, $E(\bar{W}_{\Gamma,\epsilon} )\geq E(\bar{D}_{\Gamma + \epsilon})$. Further define $U_{i, \Gamma + \epsilon}$ as before, namely
\begin{align*}
U_{i, {\Gamma} + \epsilon} &= \bar{\tau}_i + V_{i,\Gamma+\epsilon}|\eta_i| - (2\theta_{\Gamma+\epsilon}-1)\{(1 + V_{i,\Gamma + \epsilon})|\bar{\tau}_i + |\eta_i|| + (1-V_{i, \Gamma+\epsilon})|\bar{\tau}_i - |\eta_i||\}/2.\end{align*}
We now show that even the limit, $E(\bar{U}_{\Gamma + \epsilon}) > E(\bar{W}_{\Gamma,\epsilon})$.
\begin{align*} &E(\bar{U}_{\Gamma+\epsilon} ) - E(\bar{W}_{\Gamma,\epsilon} )\\ &= n^{-1}\sum_{i=1}^n\left[(\theta_{\Gamma+\epsilon} - \theta_{\Gamma})\left\{|\eta_i| - (2\theta_{\Gamma+\epsilon}-1)|\bar{\tau}_i + |\eta_i||\right\}\right.\\&\left.+ (\theta_{\Gamma} - \theta_{\Gamma + \epsilon})\left\{-|\eta_i| + (2\theta_{\Gamma+\epsilon}-1)|\bar{\tau}_i - |\eta_i||\right\}\right]\\
&= (\theta_{\Gamma+\epsilon} - \theta_{\Gamma})n^{-1}\sum_{i=1}^n\left\{2|\eta_i| + (2\theta_{\Gamma+\epsilon}-1)\left(|\bar{\tau}_i - |\eta_i|| - |\bar{\tau}_i + |\eta_i||\right)\right\}\\
&\geq 4(1-\theta_{\Gamma+\epsilon})(\theta_{\Gamma+\epsilon} - \theta_{\Gamma})n^{-1}\sum_{i=1}^n|\eta_i|,\;\;
\end{align*}
where the last line follows arguments similar to those used to prove (\ref{eq:U2moment}). In the limit, the last line is greater than or equal to $4(1-\theta_{\Gamma+\epsilon})(\theta_{\Gamma+\epsilon} - \theta_{\Gamma})C > 0$ by (\ref{eq:CLT}). Hence, if (\ref{eq:sensmodel}) holds at $\Gamma$ but a sensitivity analysis is conducted at $\Gamma+\epsilon$, $E(\bar{D}_{\Gamma+\epsilon})$ is strictly less than $E(\bar{U}_{\Gamma+\epsilon} )$ asymptotically, which is itself less than or equal to zero if the average treatment effect equals $\tau$ by Lemma \ref{lemma:expec}. 

Let $\mu_D = E(\bar{D}_{\Gamma+\epsilon} ) < 0$,  $\sigma^2_{D}/n = \var(\bar{D}_{\Gamma+\epsilon} )$, and $\nu^2_\Gamma = 4\theta_\Gamma(1-\theta_\Gamma)\nu^2$. Asymptotically, the unstudentized procedure rejects if $n^{1/2}\bar{D}_{\Gamma+\epsilon} \geq \nu_{\Gamma+\epsilon}\Phi^{-1}(1-\alpha)$ by Theorem \ref{prop:student}.
\begin{align*}
&\underset{n\rightarrow\infty}{\lim}E\{\varphi_{F}(\alpha,\Gamma+\epsilon) \mid H_N\}\\
&=\underset{n\rightarrow\infty}{\lim}\P\{n^{1/2}\bar{D}_{\Gamma+\epsilon} \geq \nu_{\Gamma+\epsilon}\Phi^{-1}(1-\alpha)\mid H_N\}\\ &= \underset{n\rightarrow\infty}{\lim}\P\{n^{1/2}(\bar{D}_{\Gamma+\epsilon} - \mu_D)/\sigma_{D} \geq (\Phi^{-1}(1-\alpha)\nu_{\Gamma+\epsilon} -n^{1/2}\mu_D)/\sigma_{D}\mid H_N\}\\
&= \underset{n\rightarrow\infty}{\lim}1 - \Phi\left\{\frac{-n^{1/2}\mu_D + \Phi^{-1}(1-\alpha)\nu_{\Gamma+\epsilon}}{\sigma_{D}}\right\} = 0, 
\end{align*}
where the last line stems from $\mu_D  < 0$ and asymptotic normality of $n^{1/2}\bar{D}_{\Gamma}$ by  Lemma \ref{lemma:kstar}.
\bibliographystyle{apalike}
\bibliography{../bibliography.bib}

\end{document}